\documentclass[a4paper]{article}

\usepackage[english]{babel}
\usepackage[utf8x]{inputenc}
\usepackage{url}
\usepackage{hyperref}
\usepackage{amsmath}
\usepackage{amssymb}
\usepackage{amsthm}
\usepackage{graphicx}
\usepackage[colorinlistoftodos]{todonotes}
\usepackage{cite}
\usepackage[blocks]{authblk}
\usepackage{mathtools}
\usepackage{bbm}
\usepackage{array}
\usepackage{cleveref}

\theoremstyle{plain}
\newtheorem{defi}{Definition}
\newtheorem{lemme}[defi]{Lemma}
\newtheorem{coro}[defi]{Corollary}
\newtheorem{prop}[defi]{Proposition}
\newtheorem{theo}[defi]{Theorem}
\theoremstyle{definition}
\newtheorem{example}[defi]{Example}
\theoremstyle{remark}
\newtheorem{remark}[defi]{Remark}

\numberwithin{defi}{section}

\def\M{\mathfrak M}

\def\nn{\mathbb{N}}
\def\zz{\mathbb{Z}}

\def\cc{\mathbb{C}}
\def\H{{\mathcal H}}
\def\h{{\mathfrak h}}

\def\tr{\mathrm{Tr}}
\def\id{\mathrm{Id}}

\makeatletter
\newcommand{\oset}[2]{%
  {\mathop{#2}\limits^{\vbox to -.5\ex@{\kern-\tw@\ex@
   \hbox{\scriptsize #1}\vss}}}}
\makeatother

\newcommand\ind{\mathbbm{1}}
\def\pp{\mathbb{P}}
\def\ee{\mathbb{E}}

\def\rhoinv{{\tau^{\mathrm{inv}}}}
\def\eps{\varepsilon}
\def\vec#1{|#1\rangle}

\def\ket#1{|#1\rangle}
\def\braket#1#2{\langle #1  , #2\rangle}
\def\ketbra#1#2{|#1\rangle\langle#2|}

\def\bD{{\partial \! D}} 

\def\mfP{{\mathfrak P}}
\def\mfPD{{\mathfrak P}^D}

\def\mfN{{\mathfrak N}}
\def\mfND{{\mathfrak N}^D}
\def\mfT{{\mathfrak T}}

\def\E{{\mathcal E}}
\def\B{{\mathcal B}}
\def\cP{{\mathcal P}}

\newcommand{\sca}[2]{\langle #1 ,#2\rangle}

\newcommand{\Hcal}{\mathcal{H}}

\newcommand{\Rcal}{\mathcal{R}}

\newcommand{\Scal}{\mathcal{S}}
\newcommand{\Dcal}{\mathcal{D}}
\newcommand{\Ecal}{\mathcal{E}}
\newcommand{\Bcal}{\mathcal{B}}

\newcommand{\proj}[2]{|#1\rangle\langle #2|}

\renewcommand{\sp}{\mathrm{Sp}\,}

\title{Passage times, exit times and Dirichlet problems for open quantum walks}
    \author{Ivan Bardet}
    \affil[1]{\small Institut Camille Jordan, Universit\'e Claude Bernard Lyon 1, 43 boulevard du 11 novembre 1918, 69622 Villeurbanne cedex, France}
    \author{Denis Bernard}
    \affil[2]{\small CNRS and Laboratoire de Physique Th\'eorique de l'Ecole Normale Sup\'erieure de Paris, PSL University, France}
    \author{Yan Pautrat}
	\affil[3]{\small Laboratoire de Math\'ematiques d'Orsay, Univ. Paris-Sud, CNRS, Universit\'e Paris-Saclay,  91405~Orsay, France}

\begin{document}

\maketitle

\begin{abstract}
    We consider open quantum walks on a graph, and consider the random variables defined as the passage time and number of visits to a given point of the graph. We study in particular the probability that the passage time is finite, the expectation of that passage time, and the expectation of the number of visits, and discuss the notion of recurrence for open quantum walks. We also study exit times and exit probabilities from a finite domain, and use them to solve Dirichlet problems and to determine harmonic measures. We consider in particular the case of irreducible open quantum walks. The results we obtain extend those for classical Markov chains.
\end{abstract}

\section{Introduction}

Open quantum walks were defined in \cite{APSS}. They are extensions of (discrete-time) Markov chains, where the process retains some amount of memory, and this memory is encoded by a quantum state. Open quantum walks are a simple model, which has stirred interest because of their various possible applications (see \cite{SinPet} and references therein for models based on open quantum walks, and \cite{ZWLJN} on the general topic of control of quantum trajectories) and interesting features and extensions (see \cite{BBToqbm,BBTbistability,PelOQRW}). They have therefore given rise to various theoretical studies, investigating e.g. ergodic properties, central limit theorems and large deviations properties (see \cite{AGS,CP1,CP2,CarvalhoGuidiLardizabal,LardSouza2}). The approach of \cite{CP1} was to give analogues for open quantum walks of notions usually associated with Markov chains, such as irreducibility and period, and to investigate their consequences.  In this article, we continue this program of studying open quantum walks in analogy with Markov chains, and investigate other notions: the probability of visiting a given site in finite time, the expected number of visits, the expected return time, and their relation with the Dirichlet problem. Some of these notions were discussed in e.g. \cite{CarvalhoGuidiLardizabal,DhahriMukhamedov}, but our study is the first systematic exploration of these concepts and their behavior for irreducible open quantum walks.
\medskip

The first standard question one may ask about Markov chains treats recurrence problems.
Let $(x_n)_n$ be a Markov chain on a discrete set $V$. For any $i$ in $V$ we define
\begin{gather*}
t_i = \inf\{n\geq1\, |\, x_n=i\}, \qquad 
n_i = \mathrm{card}\{n\geq1\, |\, x_n=i\},
\end{gather*}
The classical results (see e.g. \cite{Durrett,Norris}) concerning return times $(t_i)_{i\in V}$ and number of visits $(n_i)_{i\in V}$ imply that for any $i$ in $V$
\begin{equation} \label{eq_Classical1}
\pp_i(t_i<\infty)=1
\Leftrightarrow \ee_i( n_i)=\infty.
\end{equation}
Therefore, this equivalence allows to define the notion of recurrence using either quantity $\pp_i(t_i<\infty)$ or $\ee_i( n_i)$. In addition, if the Markov chain is irreducible,
\begin{gather}
\pp_i(t_i<\infty)<1 \mbox{ for all }i\in V, \mbox{ or } \pp_i(t_i<\infty)=1 \mbox{ for all }i\in V, \label{eq_Classical2}\\
\ee_i(n_i)<\infty \mbox{ for all }i\in V, \mbox{ or } \ee_i(n_i)=\infty \mbox{ for all }i\in V. \label{eq_Classical3}
\end{gather}
Similarly, for an irreducible Markov chain, 
\begin{equation}\label{eq_Classical4}
\ee_i(t_i)<\infty \mbox{ for all }i\in V,\mbox{ or } \ee_i(t_i)=\infty \mbox{ for any }i\in V.
\end{equation}
In addition, if the Markov chain admits an invariant probability measure $(\pi_i)_{i\in V}$, then 
\begin{equation}\label{eq_Classical5}
\ee_i(t_i)=\pi_i^{-1}<\infty \mbox{ for any }i\in V.
\end{equation}
\smallskip

The second standard question concerns exit times and exit probabilities. If $D$ is a finite subset of $V$, we define $t_{\bD}$ as its exit time
\[t_\bD=\inf \{n\geq 1\,|\, x_n\in \bD\}\]
where $\bD$ is the boundary of $D$ (we give a precise definition later on), and for $i\in D$, $j\in \bD$ define the harmonic measure at $i$ relative to $j$ by 
\begin{equation}\label{eq_harmonic_classic}
\mu_{i}^D(j)=\pp(x_{t_{\bD}}=j\,|\, x_0=i)\
\end{equation}
which represents the probability of exiting $D$ through $j$ when starting from $i$. It is known that the map $i\mapsto \mu_i(j)$ is harmonic on $D$ for any $j\in\bD$, and is an important tool in solving Dirichlet problems. In addition, the solution of a Dirichlet problem can be characterized as the minimizer of some functional, related to a Dirichlet form (see \cite{Revuz}).

In this article we investigate similar relations to \eqref{eq_Classical1}--\eqref{eq_Classical5}, and study an analogue of Dirichlet problems for open quantum walks. We also look at the notion of harmonic measures for open quantum walks. These measures, as well as the Dirichlet problems for open quantum walks, provide simple examples of non-commutative extensions of standard geometrical structures.

This article is organized as follows: in Section \ref{sec_definitionsnotations} we recall the definitions of open quantum walks and various notions, including irreducibility and harmonicity. In Section \ref{sec_passagetimesnumbervisits} we study the relation between return times and number of visits and in particular analogues for OQW of \eqref{eq_Classical1}, \eqref{eq_Classical2} and \eqref{eq_Classical3}, and discuss the literature on the subject of recurrence for open quantum walks. In Section \ref{sec_expreturntimes} we study the expected values of return times, prove an analogue of  \eqref{eq_Classical4}, and relate these expected values to invariant measures, similar to \eqref{eq_Classical5}. In Section \ref{sec_examples} we describe various examples that serve in particular as counterexamples to various possible conjectures. In Section \ref{sec_exittimesDirichletproblems} we define Dirichlet problems for open quantum walks and characterize their solutions. In Section \ref{sec_results_for_reducible_open_quantum_walks} we discuss extensions of various results to reducible open quantum walks. In Section \ref{sec_varDirichlet} we introduce Dirichlet forms for open quantum walks and use them to characterize solutions of Dirichlet problems. In every section, a non-optimal but nevertheless satisfactory result is given early in the introductory part, and the intermediate results necessary for the proof (most of which have weaker assumptions than necessary for the results stated earlier) are detailed in the rest of the section. The proofs are given in the Appendix, unless they contain elements necessary to the comprehension of the text.

\paragraph{Acknowledgements.} All three authors acknowledge the support of ANR project StoQ ``Stochastic Methods in Quantum Mechanics'', n${}^\circ$ANR-14-CE25-0003. They also want to thank St\'ephane Attal for discussions at an early stage of this project.

\section{Open quantum walks: definitions and notation}
\label{sec_definitionsnotations}

We start this section with a short presentation of open quantum walks and the associated notion of irreducibility. We follow the notation of \cite{CP1} and refer the reader to that article for more details.

We consider a Hilbert space $\H$ of the form $\H = \bigoplus_{i\in V}\mathfrak h_i$ where $V$ is a countable set of vertices, and each $\mathfrak h_i$ is a separable Hilbert space. We view $\H$ as describing the degrees of freedom of a particle constrained to move on $V$: the ``$V$-component" describes the spatial degrees of freedom (the position of the particle) while $\mathfrak h_i$ describes the internal degrees of freedom of the particle, when it is located at site $i\in V$. 

For book-keeping purposes we denote the subspace $\mathfrak h_i$ of $\H$ by $\mathfrak h_i \otimes \vec i$. Therefore, whenever a vector $\varphi\in\H$ belongs to the subspace $\mathfrak h_i$, we will denote it by $\varphi\otimes\vec i$ and drop the (implicit) assumption that $\varphi\in\mathfrak h_i$. Similarly, when an operator $A$ on $\H$ satisfies $\mathfrak h_j^\perp\subset\mathrm{Ker}\, A$ and $\mathrm{Ran}\, A\subset \mathfrak h_i$, we denote it by $A=L_{i,j}\otimes \ketbra ij$ where $L_{i,j}$ is viewed as an operator from $\mathfrak h_j$ to $\mathfrak h_i$.  This will allow us to use the same notation as in e.g. \cite{APSS, AGS, KY, LardSouza, PelOQRW}. Consistently with this notation, for $W$ a subset of $V$ we denote
\[\mathcal H_W=\bigoplus_{i\in W} \h_i\otimes \vec i.\]
and $\id_W=\sum_{i\in W} \id_{\h_i}\otimes \vec i \langle i|$. We identify $\mathcal H_W$ (respectively $\mathcal B(\mathcal H_W)$) with a subspace of $\mathcal H$ (respectively $\mathcal B(\mathcal H)$).
\smallskip

An open quantum walk (or OQW) is a map on the Banach space $\mathcal I _1(\H)$ of trace-class operators on~$\H$, given by
\begin{equation}\label{eq_OQRW}
\mathfrak M \, : \, \tau \mapsto \sum_{i,j\in V}A_{i,j}\, \tau \,A_{i,j}^*
\end{equation}
where, for any $i,j$ in $V$, the operator $A_{i,j}$ is of the form $L_{i,j}\otimes |i\rangle \langle j|$ and the operators $L_{i,j}$ satisfy
\begin{equation}\label{eq_stochastic}
\forall j\in V,\quad \sum_{i\in V}L_{i,j}^* L_{i,j}=\id_{\h_j},
\end{equation}
(this series is meant in the strong convergence sense). The operators $L_{i,j}$ represent the effect of a transition from site $j$ to site $i$, encoding both the probability of that transition and its effect on the internal degrees of freedom. Equation \eqref{eq_stochastic} therefore encodes the ``stochasticity'' of the transitions $L_{i,j}$, and immediately implies that $\tr\,\mathfrak M(\tau) = \tr\, \tau$ for any $\tau$ in $\mathcal I _1(\H)$.

Recall that an operator $X$ on $\H$ is called positive (respectively definite positive) if $\braket \varphi {X\varphi} \geq 0$ (respectively $\braket \varphi {X\varphi} > 0$) for any $\varphi\in \H\setminus\{0\}$. We define a \textit{state} on $\H$ to be a positive operator in $\mathcal I_1(\H)$ with trace one, and call a state \emph{faithful} if it is definite positive. We denote the set of states on $\H$ (respectively $\h_i$) by $\mathcal S(\H)$ (respectively $\mathcal S(\h_i)$). The map defined by \eqref{eq_OQRW} maps a state to a state. It actually has the stronger property of being trace-preserving and \emph{completely positive}, i.e. for any $n\in \nn$, $\M\otimes \id_{\mathcal B(\cc^n)}$ acting on $\mathcal I_1(\mathcal H)\otimes \mathcal B(\cc^n)$  is positive; such an operator is commonly called a \textit{quantum channel}, see e.g. \cite{wolftour}. In addition, the topological dual $\mathcal I_1(\H)^*$ can be identified with $\B(\H)$ through the linear form
\[(\tau,X)\mapsto \tr(\tau X),\]
so that the dual $\M^*$ of $\M$ acts on $\mathcal B(\mathcal H)$. By the Russo-Dye Theorem (see \cite{RussoDye}), we have the relation\footnote{Note that the two norms $\|\cdot\|$ in this relation are different, the first being the norm for operators acting on $\mathcal B({\mathcal H})$, the second the norm for operators acting on ${\mathcal H}$.} $\|\M^*\|=\|\M^*(\id_{\mathcal H})\|$, so that relation \eqref{eq_stochastic} implies that $\|\M\|=1$ as an operator on $\mathcal I_1(\H)$. 

A crucial remark is that the range of $\mathfrak M$ is a subset of the class of ``diagonal'' states, i.e. states of the form
\begin{equation} \label{eq_formstates}
\sum_{i\in V}\tau(i)\otimes \ketbra ii,
\end{equation}
where each $\tau(i)$ is in $\mathcal I_1(\h_i)$. In addition, even if $\tau$ is not diagonal, i.e. is of the form $\tau = \sum_{i,j\in V} \tau(i,j)\otimes\ketbra ij$, then $\M(\tau)$ depends only on its diagonal elements $\tau(i,i)$. Therefore, from now on, we will only consider states of the form \eqref{eq_formstates}. The action of $\M$ on such states takes the form
\begin{equation}\label{eq_Mn2}
\mathfrak M(\tau)=\sum_{i\in V}\big(\sum_{j\in V} L_{i,j}\, \tau(j)\, L_{i,j}^*\big) \otimes \ketbra ii.
\end{equation}
As argued in \cite{CP1} (see in particular section 8), a natural extension of the above framework is to encode the transition from site $j$ to site $i$ not by $\tau(j)\mapsto L_{i,j}\, \tau(j) \, L_{i,j}^*$, but by a more general, completely positive map $\tau(j)\mapsto \Phi_{i,j}\big(\tau(j)\big)$. We will not discuss these generalized open quantum walks any further.

We now describe a family of classical random processes associated with $\M$. Let $\Omega=V^{\nn}$ and for any state $\rho$ on $\H$ of the form \eqref{eq_formstates}, define on $\Omega$ a probability by defining its restrictions to $V^{n+1}$ for $n\geq 0$:
\begin{equation}\label{eq_defPtau}
\pp_\tau(i_0,\ldots,i_n)=\tr\big(L_{i_n,i_{n-1}}\ldots L_{i_1,i_0} \,\tau(i_0)\,  L_{i_1,i_0}^* \ldots L_{i_n,i_{n-1}}^* \big).
\end{equation}
Relation \ref{eq_stochastic} ensures the consistency of these restrictions, and the Daniell-Kolmo\-gorov extension Theorem ensures that $\pp_\tau$ defines a unique probability on $\Omega$. We will mostly consider initial states of the form $\rho\otimes \ketbra ii$ with $\rho \in \mathcal S(\h_i)$, i.e. with initial position $i$ and initial internal state $\rho$; for notational simplicity, the corresponding probability $\pp_\tau$ will then be denoted by~$\pp_{i,\rho}$.

We will consider two random processes $(x_n)_n$ and $(\rho_n)_n$ defined on $\omega=(i_0,i_1,\ldots) \in \Omega$ by 
\begin{equation}\label{eq_tranclassOQW}
\begin{aligned}
x_n(\omega)	&=i_n,\\
\rho_n(\omega)&= \frac{L_{i_n,i_{n-1}}\ldots L_{i_1,i_0} \,\tau(i_0)\,  L_{i_1,i_0}^* \ldots L_{i_n,i_{n-1}}^* }{\tr\big(L_{i_n,i_{n-1}}\ldots L_{i_1,i_0} \,\tau(i_0)\,  L_{i_1,i_0}^* \ldots L_{i_n,i_{n-1}}^* \big)}\,.
\end{aligned}
\end{equation}
Note that the variable $\rho_n$ is a state on $\mathfrak h_{x_n}$. Besides, the process $(x_n,\rho_n)_n$ is Markov, corresponding to the transitions defined loosely as follows: conditionally on $(x_n=j,\rho_n=\rho)$, one has
\begin{equation}\label{eq_probatransition}
(x_{n+1},\rho_{n+1})= (i, \frac{L_{i,j} \rho L_{i,j}^*}{\tr\, L_{i,j} \rho \,L_{i,j}^*}) \quad \mbox{with probability }\,\tr (\,L_{i,j} \rho \,L_{i,j}^*).
\end{equation}
Remark that $(x_n)_n$ or $(\rho_n)_n$ considered separately are not Markov processes.

Note that open quantum walks include classical Markov chains. More precisely, consider a Markov chain  $(M_n)_n$ on the vertex set~$V$, with probability $t_{i,j}$ of transition from $j$ to $i$ and initial distribution $(p_i)_{i\in V}$. Define an open quantum walk $\M$ with $\mathfrak h_i \equiv \cc$ and $L_{i,j}=\sqrt {t_{i,j}}$. If the initial state is $\tau=\sum_{i\in V} p_i \otimes \ketbra ii$ then $\M(\tau)$ is of the form
\[\M(\tau)=\sum_{i\in V} (\sum_{j\in V} t_{i,j}\,p_j)\otimes\ketbra ii.\]
Therefore, $x_0$ has the same law as $M_0$ and $x_1$ has the same law as $M_1$, etc. This open quantum walk will be called the minimal dilation of the Markov chain (because it is an OQW implementation of the Markov chain with minimal spaces $\h_i$, see \cite{CP1} for more details).

We now introduce the notion of irreducibility for open quantum walks. For $i,j$ in $V$ we call a \textit{path from $i$ to~$j$} any finite sequence $i_0,\ldots,i_\ell$ in $V$ with $\ell\geq 1$, such that $i_0=i$ and $i_\ell=j$. Such a path is said to be of length $\ell$, and we denote the length of a path $\pi$ by $\ell(\pi)$. We denote by $\cP_\ell (i,j)$ the set of paths from $i$ to $j$ of length $\ell$, and by
\[{\cP}(i,j)=\cup_{\ell\geq 1} \cP_\ell (i,j).\]
For a fixed OQW $\M$ and $\pi=(i_0,\ldots,i_\ell)$ in $\cP (i,j)$ we denote by $L_\pi$ the operator from $\mathfrak h_i$ to $\mathfrak h_j$ defined by:
$$ L_\pi=L_{i_\ell,i_{\ell-1}}\ldots L_{i_1,i_0}=L_{j,i_{\ell-1}}\ldots L_{i_1,i}.$$
\begin{defi} \label{defi_irreducibility}
An open quantum random walk $\M$ as above is irreducible if for any $i,j$ in $V$ and for any $\varphi$ in $\mathfrak h_i\setminus\{0\}$, the set
\begin{equation} \label{eq_enstotal1}
	\{L_\pi \varphi \, |\, \pi\in\cP(i,j)\}
\end{equation}
is total in $\mathfrak h_j$.
\end{defi}
This definition (which is a special case of Davies irreducibility as defined in \cite{Dav}) and its consequences were introduced in \cite{CP1}. The main consequence is that for~$\M$ an irreducible open quantum walk, the set of solutions of $\M(\tau)=\tau$ is a space of dimension at most one, and one solution is a faithful.

\begin{remark}
	There was a slight ambiguity in the definition of irreducibility as given in \cite{CP1,CP2}, where the paths ``of length zero'' $\pi=\{i\}$ with associated transition $L_{\{i\}}=\id_{\h_i}$ was allowed in \eqref{eq_enstotal1}. In other words, irreducibility was defined by the fact that for any $i,j$ in $V$ and $\varphi$ in $\mathfrak h_i\setminus\{0\}$, the set
\begin{equation} \label{eq_enstotal2}
	\{\varphi\}\cup \{L_\pi \varphi \, |\, \pi\in\cP(i,j)\}
\end{equation} (with $\pi$ of length at least one) is total. It is easy, however, to see that this second definition is equivalent to Definition \ref{defi_irreducibility}. Therefore, even though we define irreducibility by Definition \ref{defi_irreducibility}, we can still apply the results of \cite{CP1}.
\end{remark}

The following class of open quantum walks will be relevant, and in particular will have properties closer to those of (classical) Markov chains.
\begin{defi}
	We say that an open quantum walk $\mathfrak M$ is semifinite if for any~$i$ in $V$, $\dim \mathfrak h _i< \infty$. We say that it is finite if it is semifinite and $V$ is a finite set.
\end{defi}
Non-semifinite open quantum walks, i.e. those that can have an infinite number of degrees of freedom at a given site $i\in V$, can exhibit local degeneracies that will make them less interesting to us.

One of the topics of interest in the present paper is that of \textit{quantum harmonic} operators:
\begin{defi}
	An operator $A=\sum_{i\in V} A_i \otimes\ketbra ii$ with $A_i \in \B(\mathfrak h_i)$ for each~$i\in~V$ is called quantum harmonic if it satisfies $\M^*(A)=A$, or equivalently if
	\begin{equation*}
	\mbox{for any }j\in V, \mbox{ one has } A_j= \sum_{i\in V} L_{i,j}^* A_i L_{i,j}.
	\end{equation*} 
\end{defi}
Remark that any operator $\lambda \id_{\mathcal H}$, $\lambda$ in $\mathbb{C}$, is quantum harmonic. In the case of a minimal dilation of a classical Markov chain, this definition is equivalent to the classical definition of harmonicity. An immediate property of quantum harmonic operators is the following:
\begin{lemme} \label{lemme_martingale}
	Let $\M$ be an open quantum walk and $A$ a quantum harmonic operator for $\M$. Then for any initial $(x_0,\rho_0)$, the Markov chain~$(x_n,\rho_n)_n$, with transition probabilities given by Equation \eqref{eq_probatransition}, is such that $m_n=\big(\tr(\rho_n A_{x_n})\big)_n$ is a $\pp_{x_0,\rho_0}$-martingale.
\end{lemme}

Before we move on to the next section, it will be convenient to introduce some additional notation for specific paths. For any $i$ and $j$ in $V$, and for $W$ a subset of $V$, we define $\cP^{W}(i,j)$ to be the set of paths in $\cP(i,j)$ that remain in $W$ except possibly for their start- and endpoint. More precisely:
\begin{gather*}
(i_0,\ldots,i_\ell)\in\cP^{W} \Leftrightarrow(i_0,\ldots,i_\ell)\in\cP \mbox{ with } i_1,\ldots,i_{\ell-1}\in W.
\end{gather*}
We denote for any $\ell\geq 1$ by e.g. $\cP_\ell^{W}$ the subset of $\cP^{W}$ consisting of paths of length $\ell$, etc.

\section{Passage times and number of visits}
\label{sec_passagetimesnumbervisits}

In this section we consider the passage times $t_j$ to a given point $j\in V$ and the number of visits $n_j$ to that point. We define
\begin{gather*}
t_j = \inf\{n\geq1\, |\, x_n=j\}, \qquad 
n_j = \mathrm{card}\{n\geq1\, |\, x_n=j\},
\end{gather*}
Recall the standard results \eqref{eq_Classical1}, \eqref{eq_Classical2}, \eqref{eq_Classical3} in the case where the OQW is a minimal dilation of a Markov chain. The equivalence \eqref{eq_Classical1} follows from the markovianity of the process $(x_n)_n$. However, in the general OQW case, the process $(x_n)_n$ alone is not markovian and at least some of the above relations will fail. Indeed, Example \ref{ex_exampleB} below with $i=0$ and $\rho\neq\ketbra {e_1}{e_1}, \ketbra{e_2}{e_2}$ is such that $\pp_{i,\rho}(t_i<\infty)<1$ and $\ee_{i,\rho}(n_i)=\infty$, showing that \eqref{eq_Classical1} does not hold. Example \ref{ex_exampleC}, which displays an irreducible OQW, is such that $\pp_{i,\rho}(t_i<\infty)=1$ and $\ee_{i,\rho}(t_i)<\infty$ for $p<1/2$ and $i=0$, $\rho=\ketbra{e_2}{e_2}$. Consequently it shows that \eqref{eq_Classical1} may not even hold for irreducible OQWs.

We therefore have to work some more in order to obtain nontrivial connections between $\pp_i(t_i<\infty)$ and $\ee_i(n_i)$ and then to obtain universality results in the irreducible case. This will be done in the next two subsections. As a corollary of our investigation, a clear conclusion can be drawn for semifinite irreducible open quantum walks, which we give in Theorem \ref{theo_summary} below. Its proof, however, relies on all results given in those two subsections. We will finish this section with a discussion of the notion of recurrence for open quantum random walks.

\begin{theo} \label{theo_summary}
    Let $\M$ be a semifinite irreducible open quantum walk. We are in one (and only one) of the following situations:
\begin{enumerate}
    \item for any $i,j$ in $V,$ $\rho$ in $\mathcal S(\h_i)$, one has $\ee_{i,\rho}(n_j)=\infty$ and $\pp_{i,\rho}(t_j<\infty)=1$;
    \item for any $i,j$ in $V,$ $\rho$ in $\mathcal S(\h_i)$, one has $\ee_{i,\rho}(n_j)<\infty$ and $\pp_{i,\rho}(t_j<\infty)~<~1$;
    \item for any $i,j$ in $V,$ $\rho$ in $\mathcal S(\h_i)$, one has $\ee_{i,\rho}(n_j)<\infty$, but there exist $i$ in~$V$, $\rho, \rho'$ in $\mathcal S(\h_i)$ ($\rho$ necessarily non-faithful) with $\pp_{i,\rho}(t_i<\infty)~=~1$ and~$\pp_{i,\rho'}(t_i<\infty)~<~1$.
\end{enumerate}
\end{theo}
\begin{remark}\label{remark_situationsOQWirr}\,\hfill
\begin{enumerate}
\item Situation 1 is illustrated by any finite irreducible OQW, or by a $(\zz,\cc^2)$-simple OQW (see Example \ref{ex_ZC2simpleOQW}) with $L_+^*L_+=L_-^* L_-=\frac12\,\id_{\cc^2}$.
\item Situation~2 is illustrated by e.g. the minimal dilation of a transient classical Markov chain, or a $(\zz,\cc^2)$-simple OQW with, this time, $L_+^*L_+>\frac12\id_{\cc^2}$.
\item Situation 3, which of course is the most surprising in comparison with the case of classical Markov chains, is illustrated by Example \ref{ex_exampleC}.
\end{enumerate}
\end{remark}

\subsection{Passage time vs. number of visits: general results} 
\label{subsec_passage_time_vs_number_of_visits_}

In this section we investigate the quantities $\pp_{i,\rho}(t_j<\infty)$ and $\ee_{i,\rho}(n_j)$, and the relation betwen them. The proofs relating the two quantities $\pp_{i,\rho}(t_j<\infty) $ and $\ee_{i,\rho}(n_j)$ in the classical case are based on the fact that any path $\pi\in \cP(j,j)$ can be written uniquely as a concatenation of paths $\pi\in \cP^{V\setminus\{j\}}(j,j)$. We will use this simple idea in the present case. This is summarized in Proposition \ref{prop_ReptjT}:

\begin{prop} \label{prop_ReptjT}
There exists a family $(\mfP_{j,i})_{i,j\in V}$, where $\mfP_{j,i}$ is a  completely positive linear contraction from $\mathcal I_1(\h_i)$ to $\mathcal I_1(\h_j)$, such that for any $i,j$ in $V$ and any~$\rho$ in $\mathcal S(\h_i)$,
\begin{equation*} \label{eq_ReptjT}
\pp_{i,\rho}(t_j<\infty)=\tr\big(\mfP_{j,i}(\rho)\big),\qquad
\ee_{i,\rho}(n_j)= \sum_{k\geq 0} \tr \big(\mfP_{j,j}^k\circ \mfP_{j,i} (
\rho)\big),
\end{equation*}
(where the second expression is possibly $\infty$).
\end{prop}
\begin{remark}\label{remark_TiiId}\,\hfill
\begin{enumerate}
	\item The map $\mfP_{j,i}$ can be expressed by:
    \[\mfP_{j,i}(\rho)= \sum_{\pi\in\cP^{V\setminus\{j\}}(i,j)}  L_\pi \rho L_\pi^*\]
    (see the proof of Proposition \ref{prop_ReptjT} to see that this expression is meaningful).
    \item An immediate consequence of Proposition \ref{prop_ReptjT} is that
    \begin{equation*}
        \|\mfP_{j,i}\|=\sup_{\rho \in \mathcal S(\h_i)}\pp_{i,\rho}(t_j<\infty),
    \end{equation*}
    where the norm is the operator norm on $\mathcal I_1(\H)$.
    \item It is immediate from the proof that
    \begin{equation} \label{eq_esprhoti}
        \ee_{i,\rho}(\rho_{t_j}\,|\, t_j<\infty)= \frac{\mfP_{j,i}(\rho)}{\tr \,\mfP_{j,i}(\rho)}.
    \end{equation} 
\end{enumerate}
\end{remark}
\smallskip

We have the following corollary:
\begin{coro} \label{coro_ReptjT}
Let $i,j$ be in $V$.
	\begin{enumerate}
		\item One has $\pp_{i,\rho}(t_j<\infty)=1$ if and only if $\mfP_{j,i}^*(\id_{\h_j})$ has the form $\begin{pmatrix} \id & 0 \\ 0 & *\end{pmatrix}$ in the decomposition $\h_i=\mathrm{Ran}\,\rho \oplus (\mathrm{Ran}\,\rho)^\perp$. In particular, if there exists a faithful $\rho$ in $\Scal(\h_i)$ such that $\pp_{i,\rho}(t_j<\infty)=1$, then $\pp_{i,\rho'}(t_j<\infty)=1$ for any $\rho'\in\mathcal S(\h_i)$.
        \item If there exists a faithful $\rho$ in $\Scal(\h_i)$ such that $\pp_{i,\rho}(t_i<\infty)=1$, then one has $\ee_{i,\rho'}(n_i)=\infty$ for any $\rho'\in \mathcal S(\h_i)$.
		\item If there exists a faithful $\rho$ in $\Scal(\h_i)$ such that $\ee_{i,\rho}(n_j)<\infty$ and $\h_i$ is finite-dimensional, then one has $\ee_{i,\rho'}(n_j)<\infty$ for any $\rho'$ in $\mathcal S(\h_i)$.
        \item If~ $\ee_{i,\rho}(n_j)<\infty$ for every $\rho$ in $\mathcal S(\h_i)$, then there exists a completely positive linear bounded map $\mfN_{j,i}$ from $\mathcal I_1(\h_i)$ to $\mathcal I_1(\h_j)$ such that
        \begin{equation}\label{eq_Nji}
            \ee_{i,\rho}(n_j)=\tr \,\mfN_{j,i}(\rho),
        \end{equation}
        and one has the expression
    \begin{equation} \label{eq_defmfN2}
        \mfN_{j,i}(\rho)=\sum_{\pi\in \cP(i,j)}L_\pi \rho L_\pi^*.
    \end{equation}
	\end{enumerate}
\end{coro}

\begin{remark}\label{remark_apresCoro35}\,\hfill
\begin{enumerate}
	\item The second part of the first statement, and the second statement, do not hold without the faithfulness assumption, as shown by Examples \ref{ex_exampleB} and \ref{ex_exampleC}.
	\item Since $\mfP_{i,j}$ is a completely positive contraction, one has $\mfP^*_{j,i}(\id_{\h_j})\leq \id_{\h_i}$. In addition, by the Russo-Dye Theorem \cite{RussoDye}, $\|\mfP^*_{j,j}\|=\|\mfP^*_{j,j}(\id_{\h_j})\|$ so that if $\|\mfP^*_{j,j}(\id_{\h_j})\|<1$, then $\ee_{i,\rho}(n_j)<\infty$ for every $\rho$ in $\mathcal S(\h_i)$ and:
	\[\mfN_{j,i}=(\id-\mfP_{j,j})^{-1}\circ \mfP_{j,i}.\]
    \item Again, under the assumptions of point 4, an immediate consequence is
    \begin{equation*}
        \|\mfN_{j,i}\|=\sup_{\rho \in \mathcal S(\h_i)}\ee_{i,\rho}(n_j),
    \end{equation*}
    and $\mfN_{j,i}^{(\alpha)}$ introduced in Section \ref{sec_appproofspassagetimes} satisfies $\mfN_{j,i}=\lim_{\alpha\to 1}\mfN_{j,i}^{(\alpha)}$.
    \item  If $\ee_{i,\ketbra \varphi\varphi}(n_j)<\infty$ for a total set of unit vectors $\varphi$ of $\h_i$, then $\mfN_{j,i}$ can still be constructed as a densely defined (a priori unbounded) selfadjoint operator, thanks to the representation theory for closed quadratic forms (see e.g. Theorem VIII.3.13a in \cite{Kato}).
\end{enumerate}
\end{remark}

We have an easy partial converse to the third statement of Corollary \ref{coro_ReptjT} (a stronger result will be given under the additional assumption of irreducibility).

\begin{prop} \label{prop_Markovsimple}
	Let $i$ be in $V$ and assume that $\h_i$ is finite-dimensional. If $\pp_{i,\rho}(t_i<\infty)<1$ for every $\rho$ in $\mathcal S(\h_i)$, then $\ee_{i,\rho'}(n_i) <\infty$ for every $\rho'$ in $\mathcal S(\h_i)$.
\end{prop}

\begin{remark}
Proposition \ref{prop_Markovsimple} implies in particular that, if $\h_i$ is finite-dimensional and there exists $\rho$ in $\mathcal S(\h_i)$ such that $\ee_{i,\rho}(n_i)=\infty$, then there exists $\rho'$ in $\mathcal S(\h_i)$ such that $\pp_{i,\rho'}(t_i<\infty)=1$. Note that this does not necessarily hold with $\rho'=\rho$, as Examples \ref{ex_exampleB} and \ref{ex_ZC2simpleOQW} show.
\end{remark}

\subsection{The irreducible case} 
\label{subsec_passage_time_vs_number_of_visits_the_irreducible_case}


We now turn to the ``universality'' properties analogous to \eqref{eq_Classical2} and \eqref{eq_Classical3} that are expected in the irreducible case. We will prove the following:
\begin{prop} \label{prop_ExpectedNumberVisits}
	Let $\M$ be an irreducible open quantum walk and let $j$ be in~$V$. We are in one (and only one) of the following situations: 
	\begin{enumerate}
		\item for every $i$ in $V$ there exists a domain ${\mathfrak d}^n_{j,i}$, dense in $\h_i$, such that the quantity $\ee_{i,\rho}(n_j)$ is finite for any $\rho$ that has finite range contained in~${\mathfrak d}^n_{j,i}$, 
		\item for every $i$ in $V$,  for any $\rho$ in $\mathcal S(\h_i)$, the quantity $\ee_{i,\rho}(n_j)$ is infinite.
	\end{enumerate}
\end{prop}
For semifinite OQW, the picture is simpler. We have the following corollary:
\begin{coro} \label{coro_ExpectedNumberVisits}
	Let $\M$ be a semifinite irreducible open quantum walk. We are in one (and only one) of the following situations: 
	\begin{enumerate}
		\item the quantity $\ee_{i,\rho}(n_j)$ is finite for any $i,j$ in $V$ and any $\rho$ in $\mathcal S(\h_i)$,
		\item the quantity $\ee_{i,\rho}(n_j)$ is infinite for any $i,j$ in $V$ and any $\rho$ in $\mathcal S(\h_i)$.
	\end{enumerate}
\end{coro}
\begin{remark}\,\hfill
\begin{enumerate}
        \item Example \ref{ex_exampleB} shows that the above statements do not hold without the irreducibility assumption.
        \item Example \ref{ex_invariantstate} shows that any irreducible open quantum walk that admits an invariant state (and in particular a finite OQW) is in case 2 of Corollary \ref{coro_ExpectedNumberVisits}.
\end{enumerate}	
\end{remark}

The next proposition is the last ingredient to prove Theorem \ref{theo_summary}.
\begin{prop} \label{prop_EnPt}
	Let $\M$ be an irreducible open quantum walk. Assume that there exists $i,j$ in $V$ with $\mathrm{dim}\, \h_i <\infty$, $\mathrm{dim}\, \h_j <\infty$ and $\ee_{i,\rho}(n_j)=\infty$ for some $\rho$ in $\mathcal S(\h_i)$. Then $\pp_{j,\rho'}(t_j<\infty)=1$ for every $\rho'$ in $\mathcal S(\h_j)$.
\end{prop}

The proof of Theorem \ref{theo_summary} follows immediately from Corollary \ref{coro_ReptjT}, Proposition \ref{prop_Markovsimple}, Corollary \ref{coro_ExpectedNumberVisits} and Proposition \ref{prop_EnPt}.

\subsection{Notions of recurrence for open quantum walks} 
\label{subsec_notions_of_recurrence_for_open_quantum_walks}
In view of Corollary \ref{coro_ExpectedNumberVisits}, we propose the following terminology:
\begin{defi} \label{defi_TransIrr}
	A semifinite irreducible open quantum walk $\M$ is called transient if it satisfies property 1. of Corollary \ref{coro_ExpectedNumberVisits}, and recurrent if it satisfies~property 2.
\end{defi}
In other words, our classification depends on the quantity $\ee_{i,\rho}(n_i)$ being finite or infinite. Thanks to  Corollary \ref{coro_ExpectedNumberVisits}, for a semifinite irreducible open quantum walk this quantity is universal in the sense that it is either finite for all $i$ and $\rho$, or infinite for all $i$ and $\rho$. We now compare this with existing definitions of recurrence for open quantum walks and related objects.
\medskip

First of all, if the open quantum walk $\M$ is the minimal dilation of a classical Markov chain, then $\M$ is recurrent in our sense if and only if the Markov chain is recurrent in the classical sense.
\medskip

Fagnola and Rebolledo defined in \cite{FagnolaRebolledo2003} a notion of recurrence for (continuous-time) quantum dynamical semigroups. When applied to the (discrete-time) quantum dynamical semigroup $(\M^n)_n$, this definition of recurrence is that for any operator $A$ of $\mathcal B(\mathcal H)$ that satisfies $\langle \varphi, A\varphi\rangle>0$ for any $\varphi \in\mathcal H\setminus\{0\}$, the set
\begin{equation*}
	D(\mathfrak U(A))
	=\big\{\varphi=\sum_{i\in V}\varphi_i\otimes \ket i \ \mbox{ s.t. } \sum_{k\geq 0}\braket \varphi {(\M^*)^k(A)\, \varphi} < \infty\big\}.
\end{equation*}
equals $\{0\}$. We call this notion FR-recurrence. Our definition of a recurrent OQW, as can be seen from Section \ref{sec_appproofspassagetimes}, is equivalent to the fact that for any $j\in V$, $D(\mathfrak U(A_j))=\{0\}$ for $A_j=\id_{\h_j}\otimes\ketbra jj$. It is clear that if the OQW is FR-recurrent, then it is recurrent in our sense. If the OQW is not FR-recurrent, then there exists $A$ as above such that $ \sum_{k\geq 0}\braket \varphi {(\M^*)^k(A)\, \varphi} < \infty$, and if the OQW is semifinite, then for any $j$ in $V$ there exists $\lambda_j>0$ such that $\lambda_j \id_{\h_j}\otimes \ketbra jj\leq A$ and the OQW is not recurrent. Therefore, for semifinite OQWs, our notion of recurrence and FR-recurrence are equivalent.
\medskip

A series of results investigating recurrence of open quantum walks can be found in \cite{LardSouza,LardSouza2,CarvalhoGuidiLardizabal}. In particular, in \cite{LardSouza,LardSouza2,CarvalhoGuidiLardizabal}, a site $i\in V$ is called (LS)-recurrent if (in our terms) one has $\pp_{i,\rho}(t_i<\infty)=1$ for any $\rho$ in $\mathcal S(\h_i)$. The OQW is called (LS)-site-recurrent if every site $i$ in $V$ is LS-recurrent. In other words, LS-recurrence classifies sites depending on the quantity 
\begin{equation} \label{eq_qtyLS}
	\inf_{\rho\in\mathcal S(\h_i)}\pp_{i,\rho}(t_i<\infty)
\end{equation} being equal to $1$ or not. Corollary \ref{coro_ExpectedNumberVisits} shows that, for a semifinite irreducible open quantum walk, $	\inf_{\rho\in\mathcal S(\h_i)}\pp_{i,\rho}(t_i<\infty)
=1$ for some $i$ if and only if it is true for all $i$ (a fact which is not proved in \cite{LardSouza,LardSouza2,CarvalhoGuidiLardizabal}), and also that this is equivalent with recurrence in the sense of Definition \ref{defi_TransIrr}. Therefore, an irreducible semifinite OQW is LS-site-recurrent if and only if it is recurrent in our sense. Without the irreducibility assumption, point 2 of Corollary \ref{coro_ReptjT} shows that if $i$ is LS-recurrent then $\ee_{i,\rho}(n_i)=\infty$ for any $\rho$ in $\mathcal S(\h_i)$; the converse does not hold, as shown by Example \ref{ex_nouveau}. Note, however, that the quantity $\ee_{i,\rho}(n_i)$ has the advantage of being universal in $i$ and $\rho$, in the sense that (for a semifinite irreducible OQW) it is either finite for every $i$ and $\rho$, or infinite for every $i$ and $\rho$. This is not true of  $\pp_{i,\rho}(t_i<\infty)=1$, as Examples \ref{ex_exampleC} and \ref{ex_ZC2simpleOQW} show. The reason can be traced back to the fact that the set of ``diagonal'' $\varphi=\sum_{i\in V} \varphi_i\otimes\vec i$ such that $\pp_{\ketbra\varphi\varphi}(t_j <\infty)=1$, even though stable by any $L_\pi\otimes \ketbra ji$ with $\pi\in\mathcal P(i,j)$, is not a vector space, and therefore cannot be an enclosure (see the proof of Proposition \ref{prop_ExpectedNumberVisits} in Section \ref{sec_appproofspassagetimes}).

\medskip
Recently, Dhahri and Mukhamedov discussed a notion of recurrence in \cite{DhahriMukhamedov}. That notion actually concerns quantum Markov chains (objects that originate in \cite{Accardirusse,AccardiFrigerio}), and was defined in \cite{AccardiKoroliuk}. The connection with open quantum walks is established by associating a quantum Markov chain to an open quantum walk. This can be done, however, in different ways, and the property of recurrence depends on the choice of the associated quantum Markov chain. In addition, it is not clear what this notion of recurrence has to do with the properties of the random variables $(x_n)_n$. A major setback, making the associated quantum Markov chains non-canonical, is that they are constructed over the algebra $\big(\mathcal B(\bigoplus_{i\in V} \h_i)\big)^{\otimes \nn}$; a more direct connection could probably be obtained, at least when $\h_i\equiv \h$, by a construction over $\mathcal B(\h)\otimes \big(\mathcal B(\cc^V)\big)^{\otimes \nn}$, as can be done using the theory of finitely correlated states (see \cite{FNW}) which extends that of quantum Markov chains.
\medskip

Last, remark that, inspired by \cite{GVWW}, the authors of \cite{CarvalhoGuidiLardizabal} discuss an alternate notion of recurrence to a site $i\in V$. In that new notion, physically speaking, the observer does not at every time $n$ measure the position $x_n$ of the particle, but measures only whether the particle has returned to $i$ or not. Mathematically, this amounts to considering the probability space defined by $\tilde \Omega=\{0,1\}^\nn$, and probability 
\[\tilde{\pp}_{i,\rho}(\tilde \imath_1,\ldots,\tilde \imath_n)=\tr\big( \Phi_{\tilde \imath_n,\tilde \imath_{n-1}}\circ\ldots\circ \Phi_{\tilde \imath_1,1}(\rho\otimes \ketbra ii)\big)\]
where we have, for $\tau$ a state on $\mathcal H$,
\begin{gather*}
\Phi_{0,0}(\tau)= A_{i,i}\tau A_{i,i}^*\quad \Phi_{0,1}(\tau)=\sum_{j\neq i} A_{i,j}\tau A_{i,j}^* \\ \Phi_{1,0}(\tau)=\sum_{j\neq i} A_{j,i}\tau A_{j,i}^* \quad \Phi_{1,1}(\tau)=\sum_{j,k\neq i} A_{j,k}\tau A_{j,k}^*	
\end{gather*}
The new notion of recurrence is then related to the first time $\tilde t_i\geq 1$ for which the process defined by $\tilde x_n(\tilde \omega)=\tilde \imath_n$ takes the value $1$. It is easy to verify, however, that this $\tilde t_i$ has the same law under $\tilde \pp_{i,\rho}$ as $t_i$ under $\pp_{i,\rho}$, so that this alternate notion of recurrence is identical to LS-recurrence, as was noted in \cite{CarvalhoGuidiLardizabal}.


\section{Expectation of return times}
\label{sec_expreturntimes}

We now turn to results analogous to \eqref{eq_Classical4}. Our first statement is a representation result.
\begin{prop} \label{prop_expetj}
	For any $i,j$ in $V$ and $\rho$ in $\mathcal S(\h_i)$, we have
	\begin{equation}
		\ee_{i,\rho}(t_j)=\left\{
		\begin{array}{>{\displaystyle}cl}
			\sum_{\pi\in\mathcal P^{V\setminus\{j\}}(i,j)} \ell(\pi) \,\tr\, L_\pi \rho L_\pi^* & \mbox{ if }\pp_{i,\rho}(t_j<\infty)=1,\\
			+\infty & \mbox{ if }\pp_{i,\rho}(t_j<\infty)<1.
		\end{array} \right.
	\end{equation}
	If $\ee_{i,\rho}(t_j)<\infty$ for every $\rho\in \mathcal S(\h_i)$, then there exists a bounded operator $\mfT_{j,i}$ from $\mathcal I_1(\h_i)$ to $\mathcal I_1(\h_j)$ such that 
    \begin{equation} \label{eq_espti}
            \ee_{i,\rho}(t_j)=\tr \,\mfT_{j,i}(\rho).
    \end{equation}
\end{prop}
\begin{remark}\,\hfill
\begin{enumerate}
    \item In the case where $\ee_{i,\rho}(t_j)<\infty$ for every $\rho\in \mathcal S(\h_i)$ we have the expression
    \[\mfT_{j,i}(\rho)=\sum_{\pi\in\mathcal P^{V\setminus\{j\}}(i,j)} \ell(\pi) \,L_\pi \rho L_\pi^*\]
    \item We have in addition the identity (with both sides possibly $\infty$).
    \[\ee_{i,\rho}(t_i)=\frac{\mathrm d}{\mathrm d \alpha} \tr \,\mfP_{i,i}^{(\alpha)}(\rho)|_{\alpha=1}.\]
    The operators $\mfP_{i,i}^{(\alpha)}$ are defined in  Section \ref{sec_appproofspassagetimes}.
\end{enumerate}
\end{remark}

Our first relevant theorem is a universality result in the irreducible case:
\begin{theo} \label{theo_espti}
    Let $\M$  be a semifinite irreducible open quantum walk. We are in one (and only one) of the following situations:
    \begin{enumerate}
         \item for any $i$ in $V$ and $\rho$ in $\mathcal S(\h_i)$, one has $\ee_{i,\rho}(t_i)<\infty$,
         \item for any $i$ in $V$ and $\rho$ in $\mathcal S(\h_i)$, one has $\ee_{i,\rho}(t_i)=\infty$.
     \end{enumerate}
\end{theo}
Our proof uses the following intermediate universality result, similar to Proposition \ref{prop_ExpectedNumberVisits} and which can be useful in a wider setting:
\begin{prop} \label{prop_ExpectedTime}
    Let $\M$ be an irreducible open quantum walk and let $j$ be in~$V$. We are in one (and only one) of the following situations: 
    \begin{enumerate}
        \item for every $i$ in $V$, there exists a domain ${\mathfrak d}^t_{j,i}$, dense in $\h_i$, such that the quantity $\ee_{i,\rho}(t_j)$ is finite for any $\rho$ in $\Scal(\h_i)$ that has finite range contained in~${\mathfrak d}^t_{j,i}$; 
        \item for every $i$ in $V$ and $\rho$ in $\mathcal S(\h_i)$, the quantity $\ee_{i,\rho}(t_j)$ is finite.
    \end{enumerate}
\end{prop}

Our second result relates the invariant state with the expectation of return times. To state it, for $j\in V$, we define by induction for $k\in \nn$ the $k$-th return time
\[t_j^{(k)}=\inf\{n>t_j^{(k-1)}\, |\, x_n=j\}.\]
\begin{theo}\label{theo_espti2}
    Let $\M$  be a semifinite irreducible open quantum walk with an invariant state $\rhoinv=\sum_{i\in V}\rhoinv(i)\otimes \ketbra ii$.  Then we are in situation~1 of Theorem \ref{theo_espti}, and for any $i,j$ in $V$ and $\rho$ in $\mathcal S(\h_i)$, the sequence $(t_j^{(k)}/k)_{k}$ converges, with respect to $\pp_{i,\rho}$, both almost-surely and in the $\mathrm L^1$ sense, to 
    \begin{equation} \label{eq_espti2}
\ee_{i,\frac{\rhoinv(i)}{\tr \,\rhoinv(i)}}(t_i)= \big(\tr\,\rhoinv(i)\big)^{-1}.
    \end{equation}
\end{theo}
The proof is based on the K\"ummerer-Maassen ergodic Theorem and Birkhoff's ergodic Theorem. Note that Theorem 1.6 in \cite{CarvalhoGuidiLardizabal} shows a result of the same type, but with less explicit assumptions.

\section{Examples} 
\label{sec_examples}

\begin{example}\label{ex_exampleB}
	Consider the open quantum walk defined by $V=\{0,1,2\}$, with $\h_i=\cc^2$ for $i=0,1,2$ and 
	\[ L_{1,0}=\begin{pmatrix} 1& 0 \\ 0 & 0 \end{pmatrix} \qquad  L_{2,0}=\begin{pmatrix} 0 & 0 \\ 0 & 1 \end{pmatrix} \qquad L_{0,1}=\begin{pmatrix} 1& 0 \\ 0 & 1 \end{pmatrix}  \qquad L_{2,2}=\begin{pmatrix} 1& 0 \\ 0 & 1 \end{pmatrix} ,\]
	all other transitions being zero. This OQW is obviously not irreducible. Denote by $e_1$, $e_2$ the canonical basis of $\cc^2$.
	For $\rho=\begin{pmatrix}1-r & s\\\overline s&r\end{pmatrix}$ (with $r\in[0,1]$ and $|s|^2\leq r(1-r)$, so that $r=1$ if and only if $\rho=\ketbra{e_2}{e_2}$) one has $\pp_{0,\rho}(t_0<\infty)=1-r$, and
	$\ee_{0,\rho}(n_0)= 0$ if $r=1$, and $\infty$ otherwise. One therefore has 
	\begin{alignat*}{4}
		&\pp_{0,\rho}(t_0<\infty)=1,&  &\ \ee_{0,\rho}(n_0)=\infty,& &\ \ee_{0,\rho}(t_0)=2& &\mbox{ for }\rho=\ketbra{e_1}{e_1};\\
		&\pp_{0,\rho}(t_0<\infty)=0,&  &\ \ee_{0,\rho}(n_0)=0,& &\ \ee_{0,\rho}(t_0)=\infty& &\mbox{ for }\rho=\ketbra{e_2}{e_2};  \\
		&\pp_{0,\rho}(t_0<\infty)\in]0,1[,& &\ \ee_{0,\rho}(n_0)=\infty,& &\ \ee_{0,\rho}(t_0)=\infty& &\mbox{ otherwise}.
	\end{alignat*}
	In this case, it is easy to compute the operator $\mfP_{0,0}$:
	\[\mfP_{0,0}(\rho)=\begin{pmatrix} 1& 0 \\ 0 & 0 \end{pmatrix} \rho \begin{pmatrix} 1& 0 \\ 0 & 0 \end{pmatrix}.\]
	Therefore, $\mfP_{0,0}^{*\,k}(\id_{\h_0})=\begin{pmatrix} 1& 0 \\ 0 & 0 \end{pmatrix}$ for any $k\geq 1$, so that loosely speaking, one has $\sum_{k\geq 0}\mfP_{0,0}^{*\,k}(\id_{\h_0})=\begin{pmatrix} \infty& 0 \\ 0 & 0 \end{pmatrix}$, consistently with Proposition \ref{prop_ReptjT}.
\end{example}

\begin{example}\label{ex_exampleC}
	Consider the open quantum walk defined by $V=\{0,1,2,\ldots\}$ with $\h_0=\cc^2$ and $\h_i=\cc$ for $i>0$, and transition operators
	\[ L_{0,0}=\begin{pmatrix} 0& 1 \\ 0 & 0 \end{pmatrix} \qquad  L_{1,0}=\begin{pmatrix} 1 & 0 \end{pmatrix} \qquad  L_{0,1}=\sqrt{p/2\,}\begin{pmatrix} 1 \\ 1 \end{pmatrix}\]
	and  $L_{i,i+1}=\sqrt p$, $L_{i+1,i}=\sqrt q$ for $i\geq 1$, with $p+q=1$ (all other transitions being zero). This OQW is semifinite and irreducible, independently of the value of $p$. However, it is a simple exercise to see that, depending on the value of $p$, one has different behaviors for $\pp_{0,\rho}(t_0<\infty)$ and $\ee_{0,\rho}(n_0)$: defining for $p\geq 1/2$ the quantity $\lambda=\frac{8p^3-8p^2+6p-1}{4p(2p-1)}\in[0,\infty]$, one has
	\begin{itemize}
		\item for $p\geq 1/2$ one has
		\begin{alignat*}{2}
				\pp_{0,\rho}(t_0<\infty)&=1 & & \quad \forall \rho \in \mathcal S(\h_i)\\
				\ee_{0,\rho}(n_0)&=\infty & &\quad \forall \rho \in \mathcal S(\h_i);\\
				\ee_{0,\rho}(t_0)&=r+2\lambda (1-r) & & \quad \mbox{for } \rho=\begin{pmatrix}1-r & s\\\overline s&r\end{pmatrix}.
		\end{alignat*}
		\item for $p<1/2$ one has
		\begin{alignat*}{2}
		\pp_{0,\rho}(t_0<\infty)=1 \mbox{ and } \ee_{0,\rho}(t_0)=1 &\quad  \mbox{if } \rho=\ketbra{e_2}{e_2}\\
		\pp_{0,\rho}(t_0<\infty)<1 \mbox{ and } \ee_{0,\rho}(t_0)=\infty&\quad \mbox{if } \rho \neq \ketbra{e_2}{e_2},\\
		\ee_{0,\rho}(n_0)<\infty &\quad \forall \rho \in \mathcal S(\h_i).
		\end{alignat*}
	\end{itemize}
	Here again it is easy to compute the operator $\mfP_{0,0}$: for any $\rho$ in $\mathcal I_1(\h_0)$ we have
	\[\mfP_{0,0}(\rho)=\begin{pmatrix} 0&1\\0 &0\end{pmatrix} \rho \begin{pmatrix} 0&0\\1 &0\end{pmatrix}+ \frac12 \inf(p/q,1) \begin{pmatrix} 1&0\\1 &0\end{pmatrix} \rho \begin{pmatrix} 1&1\\0 &0\end{pmatrix}.\]
	In particular, 
	for $p\geq 1/2$, one has  $\mfP_{0,0}^{*k}(\id_{\h_0})=\id_{\h_0}$, and for $p<1/2$ one has $\mfP_{0,0}^{*k}(\id_{\h_0})= \begin{pmatrix} u_{k+1} & 0 \\ 0 & u_k\end{pmatrix}$ with $u_k\to 0$ exponentially fast, so that $\mfP_{0,0}^{*k}(\id_{\h_0})$ is summable. This is consistent with Corollary \ref{coro_ReptjT}. In addition, for $p\geq 1/2$ one has
	\[\mfT_{0,0}(\rho)= \begin{pmatrix} 0&1\\0 &0\end{pmatrix} \rho \begin{pmatrix} 0&0\\1 &0\end{pmatrix} + \lambda \begin{pmatrix} 1&0\\1 &0\end{pmatrix} \rho \begin{pmatrix} 1&1\\0 &0\end{pmatrix}.\]
\end{example}

\begin{example}\label{ex_invariantstate}
    In the case of an irreducible open quantum walk which admits a faithful invariant state (for example a finite irreducible open quantum walk), the K\"ummerer-Maassen Theorem (proved originally in \cite{KuMa}, see \cite{Lim2010} for an infinite-dimensional extension and \cite{CP1} for an application to the case of open quantum walks) immediately implies that, for any initial position $i$ in~$V$ and any state $\rho$ in $\mathcal S(\mathfrak h_i)$, any point $j$ is almost-surely visited infinitely often:
    \[\pp_{i,\rho}(n_j=\infty)=1.\]
    A fortiori, one has $\pp_{i,\rho}(t_j<\infty)=1$ and $\ee_{i,\rho}(n_j)=\infty$, and therefore the OQW is always recurrent. This is the same as in the classical case, where an irreducible Markov chain on a finite set is always recurrent.
\end{example}

\begin{example}\label{ex_nouveau}
    Consider the open quantum walk with $V=\{0,1,2,3\}$, $\h_{0}=\cc$ and $\h_i=\cc^2$ for $i=1,2,3$.
    \begin{gather*}
    L_{0,0}=1, \qquad L_{0,1}=\begin{pmatrix}1/2 & 0 \end{pmatrix}\\
    L_{2,1}=\begin{pmatrix} \sqrt 3 /2 & 0 \\ 0 & 1 \end{pmatrix} \qquad L_{1,2}=\begin{pmatrix} 0 & 0 \\ 0 & 1 \end{pmatrix}\\
    L_{3,2}=\begin{pmatrix} 1 & 0 \\ 0 & 0 \end{pmatrix} \qquad L_{2,3}=\begin{pmatrix} 0 & 1 \\ 1 & 0 \end{pmatrix},
    \end{gather*}
all other transitions $L_{i,j}$ being zero.
One checks by examination that, starting from $i=1$ and $\rho=\begin{pmatrix}1-r & s\\\overline s&r\end{pmatrix}$: 
    with probability $(1-r)/2$ the first step goes to $0$ and then the walk stays there and, with probability $(1+r)/2$ the first step goes to $2$ and then, after a finite number of steps, goes back and forth between $1$ and $2$. Therefore, for any $\rho$, one has $\pp_{1,\rho}(t_0<\infty)=(1+r)/2$ but $\ee_{1,\rho}(n_0)=\infty$. Again it is easy to compute $\mfP_{1,1}$:
    \[\mfP_{1,1}(\rho)=\begin{pmatrix} 0 & 0 \\ 0 & 1 \end{pmatrix} \rho \begin{pmatrix} 0 & 0 \\ 0 & 1 \end{pmatrix}  + \begin{pmatrix} 0 & 0 \\ \sqrt 3 /2 & 0 \end{pmatrix}  \rho \begin{pmatrix} 0 & \sqrt 3 /2\\ 0  & 0 \end{pmatrix}.\]
    One then has $\mfP_{1,1}^{*\,k}(\id_{\h_0})=\begin{pmatrix} 3/4 & 0 \\ 0 & 1        
    \end{pmatrix}$, so that loosely speaking, $\sum_k \mfP_{1,1}^{*\,k}(\id_{\h_0})= \begin{pmatrix} \infty & 0 \\ 0 & \infty \end{pmatrix}$, consistently with Corollary \ref{coro_ReptjT}.
\end{example}

\begin{example} \label{ex_ZC2simpleOQW}
    We consider now the case of (space) homogeneous nearest-neighbor random walks on $V=\zz$ with $\h_i\equiv\h=\cc^2$ for all $i\in \zz$. This OQW is entirely determined by two operators $L_+$ and $L_-$ on $\cc^2$ satisfying $L_+^* L_+ + L_-^* L_-=\id_{\h}$. We call such an open quantum walk a $(\zz,\cc^2)$-simple OQW. It is proven in \cite{CarvalhoGuidiLardizabal} that:
    \begin{enumerate}
        \item if $\sp L_+^*L_+=\sp L_-^*L_-=\{1/2\}$ then $\pp_{i,\rho}(t_i<\infty)~=~1$ for any $i$ and $\rho$,
        \item if $\pp_{i,\rho}(t_i<\infty)=1$ for any $i$ and $\rho$ with $L_+$, $L_-$ normal, then $\sp L_+^*L_+=\sp L_-^*L_-=\{1/2\}$.
    \end{enumerate}
    With the tools developed in this section, we can recover the first point and make the second more precise. First, if $\sp L_+^*L_+=\sp L_-^*L_-=\{1/2\}$ then
    \[\mfP^*_{i,i}(\id_{\h})= \sum_{\pi\in\cP^{V\setminus\{i\}}(i,i)} \big(\frac12)^{\ell(\pi)} \,\id_{\h}\]
    which, by the results on (classical) simple random walks, is just $\id_\h$, and by Corollary \ref{coro_ReptjT}, $\pp_{i,\rho}(t_i<\infty)=1$ for any $\rho$. Second, assume that $L_+$ and $L_-$ are normal, and consider a diagonal basis for $L_+^* L_+$ (and therefore for $L_-^* L_-$). In this basis, one has
    \begin{equation} \label{eq_formeLpLm}
        L_+^* L_+ = \begin{pmatrix} p_1 & 0 \\ 0 & p_2 \end{pmatrix}\qquad
    L_-^* L_- = \begin{pmatrix} q_1 & 0 \\ 0 & q_2 \end{pmatrix}
    \end{equation}
    with $p_k+q_k=1$, $k=1,2$. It is then easy to show that if a path $\pi$ is made of $n_+(\pi)$ ``up'' steps, and $n_-(\pi)$ ``down'' steps, then 
    \[L_\pi^* L_\pi = (L_+^* L_+)^{n_+(\pi)} (L_-^* L_-)^{n_-(\pi)}= \begin{pmatrix} p_1^{n_+(\pi)}q_1^{n_-(\pi)} & 0 \\ 0 & p_2^{n_+(\pi)}q_2^{n_-(\pi)}\end{pmatrix},\]
    and using again standard results on simple random walks we have
    \[\mfP^*_{0,0}(\id_{\h})=\begin{pmatrix} \inf(2p_1,2q_1) & 0 \\ 0 & \inf(2p_2,2q_2) \end{pmatrix}.\]
    Therefore, if $L_+$ and $L_-$ are normal, then:
    \begin{itemize}
         \item if $p_1=p_2=1/2$, then $\pp_{i,\rho}(t_i<\infty)=1$ for any $i$ and $\rho$, and therefore $\ee_{i,\rho}(n_i)=\infty$;
         \item if $p_1$ and $p_2$ are both $\neq 1/2$, then $\pp_{i,\rho}(t_i<\infty)<1$ for any $i$ and $\rho$, and therefore $\ee_{i,\rho}(n_i)<\infty$;
         \item if e.g. $p_1=1/2$ and $p_2\neq 1/2$ then $\pp_{i,\rho}(t_i<\infty)=1$ for $\rho=\ketbra {e_1}{e_1}$ and $<1$ otherwise. If furthermore the OQW is irreducible (see Proposition 6.12 in \cite{CP2} for a necessary and sufficient condition), then by Theorem \ref{theo_summary} one has $\ee_{i,\rho}(n_i)=\infty$ for any $i$ and $\rho$.
     \end{itemize}  
\end{example}

A natural question is what happens when we drop the assumption of normality for $L_+$ and $L_-$. We can still assume the form \eqref{eq_formeLpLm}; if $\sup(p_1,p_2)<1/2$ or $\inf(p_1,p_2)>1/2$ and $L_+$, $L_-$ do not have an eigenvector in common, then Theorem 5.4 in \cite{CP2} implies that the process satisfies a law of large numbers $x_n/n \to m\neq 0$ almost-surely, and satisfies a large deviations principle with respect to any $\pp_{i,\rho}$. This is enough to show that $\ee_{i,\rho}(n_i)<\infty$ for any $i$ and~$\rho$. On the other hand, if e.g. $p_1>1/2$ and $p_2<1/2$ then we can still have $\pp_{i,\rho}(t_i<\infty)=1$ for any $i$ and $\rho$: consider (as suggested by \cite{CarvalhoGuidiLardizabal}) the case
\[L_+=\frac1{\sqrt 2}\begin{pmatrix} 1 & 1 \\ 0 & 0 \end{pmatrix} \qquad L_-=\frac1{\sqrt 2}\begin{pmatrix*}[r] 0& 0 \\ 1 & -1 \end{pmatrix*}\]
where $\sp L_+^*L_+ = \sp L_-^* L_-=\{0,1\}$. By Proposition 6.12 in \cite{CP2}, this open quantum walk is irreducible. In addition, for any $\pi=(i_0,\ldots,i_\ell)$, denoting $\eps=i_1-i_0$, one shows that $2^{\ell(\pi)/2}\,L_\pi$ equals
\[
    \pm \begin{pmatrix} 1 & \eps \\ 0 & 0 \end{pmatrix} \mbox{ if }i_\ell-i_{\ell-1}=+1, \qquad \pm \begin{pmatrix}  0 & 0\\ 1 & \eps \end{pmatrix}  \mbox{ if }i_\ell-i_{\ell-1}=-1.
\]
We can therefore compute, again using results for simple random walks,
\[\mfP_{0,0}(\rho) = \frac12 \begin{pmatrix*}[r] 1 & -1 \\ 0 & 0\end{pmatrix*} \rho \begin{pmatrix*}[r] 1 & 0 \\ -1 & 0\end{pmatrix*} + \frac12 \begin{pmatrix*}[r] 0 & 0 \\ 1 & 1\end{pmatrix*} \rho \begin{pmatrix*}[r] 0 & 1 \\ 0 & 1\end{pmatrix*}.\]
We therefore have $\mfP^*_{0,0}(\id_{\h_0})=\id_{\h_0}$, so that $\pp_{i,\rho}(t_j<\infty)=1$ and $\ee_{i,\rho}(n_j)=\infty$ for any $i,j$ and $\rho$. In addition, $\ee_{i,\rho}(t_i)=\infty$ for any $i$ and $\rho$.

\section{Exit times and Dirichlet problems on finite domains}
\label{sec_exittimesDirichletproblems}

In this section, we consider a finite subset $D$ of $V$ and study whether, conditionally on starting with $x_0$ in $D$, the position process $(x_n)_n$ reaches the boundary $\bD$ of $D$ (which we define below) in finite time. We then study the related problem of solving Dirichlet problems of the type $(\id-\M^*)(Z)_i=A_i$ for every $i$ in $D$, with a boundary condition $Z_j=B_j$ for $j\in \bD$. Before we start, however, let us discuss shortly the Dirichlet problem on $V$. We consider an irreducible open quantum walk $\M$, fix $A=\sum_{i\in V} A_i\otimes \ketbra ii$ with $A_i$ in $\Bcal(\h_i)$ for all $i$, and look for a solution $Z$ of the equation $(\id-\M^*)(Z)=A.$
As in the classical case (see e.g. \cite{LawlerLimic}), the form of the solution differs, depending on the recurrence or transience of the OQW. We give here only a simple result in the transient case. We define the Dirichlet problem on $V$ with data $A$ as the following equation with unknown $Z$:
\begin{equation} \label{eq_pbDirichletV}
    (\id-\M^*)(Z)=A.
\end{equation}
The operators $\mfN_{j,i}$ as defined in \eqref{eq_defmfN2} play a central role in this section.

\begin{prop}\label{prop_probDirichletV}
    Let $\M$ be an open quantum walk such that $\ee_{i,\rho}(n_j)<\infty$ for any $i$, $j$ in $V$ and $\rho$ in $\mathcal S(\h_i)$. If we assume that $A=\sum_{i\in V} A_i\otimes \ketbra ii$ is such that for any~$i$ in $V$, $\sum_{j\in V}\|\mfN_{j,i}^*\big(A_j\big)\|<\infty$, then the operator
    \begin{equation} \label{eq_solutionpbDirichletV}
        Z=A+\sum_{i\in V} \big(\sum_{j\in V}\mfN_{j,i}^*(A_j)\big)\otimes \ketbra ii
    \end{equation}    
    satisfies \eqref{eq_pbDirichletV}.
    If in addition $\M$ is irreducible, then any two solutions of \eqref{eq_pbDirichletV} differ only by an operator $\lambda \id_{\mathcal H}$. 
\end{prop}
The form of $Z$ can be guessed by analogy with the classical case, so that this result is obtained by direct computation.
\smallskip

We will give analogous results for the Dirichlet problem on a bounded domain. We start by defining precisely the boundary $\bD$ of $D$ relative to an open quantum walk $\M$:
\begin{align*}
\bD &= \{i\in V\setminus D \, |\, \exists j\in D  \mbox{ with } L_{i,j}\neq 0\}.
\end{align*}
We say that $Z=\sum_{i\in V} Z_i\otimes \ketbra ii$ is a solution to the Dirichlet problem on $D$ with data $A$ and boundary condition $B$ if
\begin{equation}\label{eq_pbDirichletD}
    \left\{
    \begin{array}{r@{\ = \ }l@{\ }l}
        (\id-\M^*)(Z)_i& A_i & \mbox{ for }i\in D \\
        Z_j&B_j & \mbox{ for }j\in \bD.
    \end{array}
    \right.
\end{equation}
A key step in order to solve explicitly this equation will be to prove that the exit time for $D$, defined as
\[ t_{\bD} = \inf\{n\in \nn \, |\, x_n\in \bD \},\]
is $\pp_{i,\rho}$-almost-surely finite for any $i$ in $D$ and $\rho$ in $\mathcal S(\h_i)$.
Our main results are summarized in the following statement:
\begin{theo}\label{theo_exittimesDirichlet}
    Let $\M$ be a semifinite irreducible open quantum walk and let $D$ be a finite subset of $V$ such that $\bD\neq \emptyset$. Then for any $i$ in $D$ and any state $\rho$ on $\mathfrak h_i$,
    $$ \pp_{i,\rho}( t_{\bD} <  +\infty)=1.$$
    In addition, for any $A=\sum_{i\in D} A_i\otimes \ketbra ii$ and $B=\sum_{j\in \bD} B_j\otimes \ketbra jj$, the  Dirichlet problem \eqref{eq_pbDirichletD} has a solution,
    and any two solutions of \eqref{eq_pbDirichletD} differ by an operator with support in $\mathcal H_{V\setminus(D\cup \bD)}$.
\end{theo}
The steps in order to prove this, and related results, are described in Subsection \ref{subsec_exit_times_irreducible_case} and Subsection \ref{subsec_dirichlet_problems_irreducible_case}.

\subsection{Exit times: the irreducible case} 
\label{subsec_exit_times_irreducible_case}

We now focus on the particular case where the OQW is irreducible. Our first technical result, which plays an analogous role to Proposition \ref{prop_ReptjT}, is the following:
\begin{prop} \label{prop_sigmaDfinite}
    Let $\M$ be an open quantum walk and let $D$ be a finite subset of $V$ such that $\bD\neq \emptyset$. Then there exists a family $(\mfPD_{j,i})_{i\in D,j\in D\cup\bD}$, where $\mfPD_{j,i}$ is a completely positive linear contraction from $\mathcal I_1(\h_i)$ to $\mathcal I_1(\h_j)$, such that each 
    \[\mfPD_{i}=\sum_{j\in \bD}\mfPD_{j,i}\]
    is again a completely positive linear contraction from $\mathcal I_1(\h_i)$ to $\mathcal I_1(\h_\bD)$. Moreover, for any $i$ in $D$, $j$ in $D\cup\bD$ and any $\rho$ in $\mathcal S(\h_i)$, one has
    \[\pp_{i,\rho}(t_j\leq t_\bD<\infty)=\tr\, \mfP^D_{j,i}(\rho), \qquad \pp_{i,\rho}(t_\bD<\infty)=\tr\, \mfP^D_{i}(\rho).\]
\end{prop}

\begin{remark} Again, a byproduct of our proof will be the expression
    \[\mfP_{j,i}^D(\rho)= \sum_{\pi\in\cP^{D\setminus\{j\}}(i,j)}  L_\pi \rho L_\pi^*.\]
    We also obtain the relations
    \begin{align} 
        \ee_{i,\rho}(\rho_{t_j}\,|\, {t_j \leq t_\bD})=\frac{\mfPD_{j,i}(\rho)}{\tr\big(\mfPD_{j,i}(\rho)\big)} \label{eq_erhotj1} \\
        \ee_{i,\rho}(\rho_{t_\bD})=\frac{\mfPD_{i}(\rho)}{\tr\big(\mfPD_{i}(\rho)\big)}. \label{eq_erhotj2}
    \end{align}
\end{remark}
The first part of Theorem \ref{theo_exittimesDirichlet} is shown in the following proposition. Apart from having its own interest, it will be a key step in solving the Dirichlet problem:
\begin{prop} \label{prop_sortieps}
    Let $\M$ be an irreducible open quantum walk and let $D$ be a finite subset of $V$ such that $\bD\neq \emptyset$. Then, for any $i$ in $D$ such that $\mathrm{dim}\,\h_i<\infty$ and any state $\rho$ on $\mathfrak h_i$, one has
    \[\pp_{i,\rho}( t_{\bD} <  +\infty)=1.\]
\end{prop}

The main consequence of Proposition \ref{prop_sortieps} is the following:
\begin{lemme} \label{lemme_mfND}
    Under the assumptions of Proposition \ref{prop_sortieps}, for any $j$ in $D$ such that $\mathrm{dim}\,\h_j<\infty$, the map $\mfPD_{j,j}$ has norm $\|\mfPD_{j,j}\|<1$. For any $i$ in $D$ one can define $\mfND_{j,i}=(\id - \mfPD_{j,j})^{-1}\circ \mfPD_{j,i}$. Then, defining
    \[n_j^D= \mathrm{card}\{n \leq t_\bD \, |\, x_n=j\},\] 
    one has for any $\rho$ in $\mathcal S(\h_i)$ the identity
    \begin{equation} \label{eq_eeirho}
        \ee_{i,\rho}(n_j^D)=\tr\, \mfND_{j,i}(\rho ).
    \end{equation}
\end{lemme}

\begin{remark}
    Under the assumptions of Lemma \ref{lemme_mfND}, the operator $\mfND$ satisfies
    \begin{equation} \label{eq_mfND}
        \mfND_{j,i}(\rho)= \sum_{\pi\in\mathcal P^D(i,j)} L_\pi \rho L_\pi ^*,
    \end{equation} 
    which shows in particular the obvious relation $\mfND_{j,i}=\mfPD_{j,i}$ for $i\in D$, $j\in \bD$.
\end{remark}

\subsection{Dirichlet problems on $D$: the irreducible case} 
\label{subsec_dirichlet_problems_irreducible_case}


We now turn to the Dirichlet problem on a finite domain $D$. Recall  that $Z=\sum_{i\in V} Z_i\otimes \ketbra ii$ is said to be a solution to the Dirichlet problem on $D$ with data $A$ and boundary condition $B$ if it is a solution of Equation \eqref{eq_pbDirichletD}, which we recall:
\begin{equation*} 
    \left\{
    \begin{array}{r@{\ = \ }l@{\ }l}
        (\id-\M^*)(Z)_i& A_i & \mbox{ for }i\in D \\
        Z_j&B_j & \mbox{ for }j\in \bD.
    \end{array}
    \right.
\end{equation*}
The solution to these equations has a very simple form, now that we have introduced the operators~$\mfND_{j,i}$ and $\mfPD_{j,i}$.
\begin{prop} \label{prop_pbDirichletD}
    Let $\M$ be an irreducible, semifinite open quantum walk and let $D$ be a finite subset of $V$ such that $\bD\neq \emptyset$. For any
    \[A=\sum_{i\in D} A_i\otimes \ketbra ii\quad \mbox{ and } \quad B=\sum_{j\in \bD} B_j\otimes \ketbra jj,\]
    the operator 
    \[Z= A+B+\sum_{i\in D}\Big(\sum_{j\in D} \mfN^{D\,*}_{j,i}(A_j)+\sum_{j\in \bD} \mfP^{D\,*}_{j,i}(B_j)\Big)\otimes \ketbra ii\]
    is well-defined, and is a solution to the Dirichlet problem \eqref{eq_pbDirichletD}. Any two solutions of \eqref{eq_pbDirichletD} differ by an operator with support in $\mathcal H_{V\setminus(D\cup \bD)}$.
\end{prop}

\subsection{Harmonic measures: the irreducible case} 
\label{subsec_harmonic measures}

Given a finite subdomain $D$ of $V$, the harmonic measure quantifies the probability for an OQW starting in $D$ to escape from $D$ by a given point of its boundary $\bD$. It is intimately related to the Dirichlet problem.

\begin{defi}
Let $D$ be a finite subset of $V$ such that $\bD\neq \emptyset$.
Let $\M$ be an irreducible and semifinite open quantum walk, conditioned to start at $i$ in $D$ with initial state $\rho$ in $\Scal(\h_i)$. Let $j$ be in $\bD$. Recall the definition of the stopping time $t_{\bD}=\mathrm{inf}\{n\in\mathbb{N}\,|\,x_n\in \bD\}$.
The harmonic measure at $j$ relative to $i$ and $\rho$ is defined as
\[ \mu^D_{i,\rho}(j)= \mathbb{P}_{i,\rho}\left(x_{t_{\bD}}=j\right).\]
\end{defi}

Recall that for an irreducible and semifinite OQW, the escape time $t_{\bD}$ is finite with probability $1$.
The previous propositions directly imply the following:
\begin{prop} 
The harmonic measure is linear in $\rho$. More precisely,
\begin{equation}  
\mu^D_{i,\rho}(j) = \mathrm{Tr}( \mfPD_{j,i}(\rho) ),
\end{equation}
with $\mfPD_{j,i}(\rho)=\sum_{\pi\in\mathcal{P}^D(i,j)} L_\pi\rho L^*_\pi$. Moreover,
\[ \mathbb{E}_{i,\rho}\left(\rho_{t_{\bD}}| x_{t_{\bD}}=j\right) = \frac{\mfPD_{j,i}(\rho)}{\mathrm{Tr}( \mfPD_{j,i}(\rho) )}.\]
\end{prop}

Of course by Proposition \ref{prop_sortieps}, $\sum_{j\in \bD} \mu^D_{i,\rho}(j) =1$, as we assume the OQW to be irreducible. 

\begin{remark}
The connection with the Dirichlet problem is twofold, as usual. 
First, by linearity in $\rho$ let us write the harmonic measure as :
\[ \mu^D_{i,\rho}(j) = \mathrm{Tr}(\mathfrak{I}^D_{j,i}\, \rho), \]
with $\mathfrak{I}^D_{j,i}= {\mfPD_{j,i}}^*(\mathrm{Id}_{\mathfrak{h}_j})=\sum_{\pi\in\mathcal{P}^D(i,j)}  L^*_\pi L_\pi$. Let  $I^D_j\in\mathcal{B}(\mathcal{H})$ be defined by
\begin{equation}
I^D_j=\sum_{i \in D} \mathfrak{I}^D_{j,i}\otimes |i\rangle\langle i| + \mathrm{Id}_{\mathfrak{h}_j}\otimes |j\rangle\langle j|.
\end{equation}
Then $I^D_j$ is quantum harmonic in $D$ (i.e $ (\id-\M^*)(I^D_j)_k=0$ for $k\in D$) with boundary condition $\mathrm{Id}_{\mathfrak{h}_j}\otimes |j\rangle\langle j|$ on $\bD$. Furthermore, one has $\sum_{j\in \bD} I^D_j = \mathrm{Id}_{D \cup \bD}$, so that $I^D_j$ may be viewed as a non-commutative quantum analogue of a harmonic measure.
\end{remark}

This link with the Dirichlet problem potentially gives an alternative way to evaluate the harmonic measure. Indeed, assuming the harmonic measure to be linear in $\rho$ (as expected from quantum mechanics), it is then fully determined by solving a Dirichlet problem. Suppose (as we actually proved) that $\mu^D_{i,\rho}(j)$ is linear in $\rho$, and let us write $\mu^D_{i,\rho}(j) = \mathrm{Tr}(\mathfrak{I}^D_{j,i}\, \rho)$ (without knowing the explicit expression of $\mathfrak{I}^D_{j,i}$). Then conditioning on the first step of the OQW proves that $I^D_j=\sum_{i \in D} \mathfrak{I}^D_{j,i}\otimes |i\rangle\langle i| + \mathrm{Id}_{\mathfrak{h}_j}\otimes |j\rangle\langle j|$ is quantum harmonic in $D$ with the appropriate boundary conditions. 
Furthermore, if $I_j^D$ is quantum harmonic in $D$ with boundary condition $\mathrm{Id}_{\mathfrak{h}_j}\otimes |j\rangle\langle j|$ then, by Lemma \ref{lemme_martingale}, $(m_j^D)_n=\mathrm{Tr}((I_j^D)_{x_n}\rho_n)$ stopped at  $n=t_{\bD}$ is a $\mathbb{P}_{i,\rho}$-martingale. The optimal sampling Theorem then yields
\[ \mathrm{Tr}(\mathfrak{I}^D_{j,i}\rho)= \mathbb{E}_{i,\rho}\big((m_j^D)_{t_{\bD}}\big)= \mathbb{P}_{i,\rho}\big(x_{\bD}=j\big).\]

\section{Results for reducible open quantum walks} 
\label{sec_results_for_reducible_open_quantum_walks}

In this section, we collect some results regarding passage times, number of visits and exit times for reducible open quantum walks. We begin by recalling some results regarding decompositions of reducible open quantum walks from \cite{CP1}.

We fix an open quantum walk $\M$. By Proposition 7.11 in \cite{CP1}, there exists an orthogonal decomposition of $\H$
\begin{equation} \label{eq_IrreducibleDecomposition} \H = \mathcal D \oplus \bigoplus_{\kappa\in K} \H^\kappa\end{equation}
where we have
\begin{equation}\label{eq_defR}
\bigoplus_{\kappa\in K} \H^\kappa = \sup\{\mathrm{supp}\,\rho\, |\, \rho \mbox{ a }\mathfrak M\mbox{-invariant state}\}.
\end{equation}
We denote the space \eqref{eq_defR} by $\Rcal$; we also let $\mathcal D = \mathcal R^\perp$. The restriction $\mathfrak M^\kappa$ of $\mathfrak M$ to $\mathcal I_1(\H^{\kappa})$ is an irreducible OQW, and $\Dcal=\{0\}$ if and only if $\M$ has a faithful invariant state. In addition, each $\H^\kappa$ is an enclosure, i.e. has a decomposition
$$ \H^\kappa = \bigoplus_{i\in V} \mathfrak h^\kappa_i$$
where every $\mathfrak h^\kappa_i$ is a subspace of $\mathfrak h_i$, and for any $i$, $j$ in $V$ one has
$ L_{i,j}\,\mathfrak h^\kappa_j \subset \mathfrak h^\kappa_i$. The decomposition \eqref{eq_IrreducibleDecomposition} is non-unique, as is discussed in sections 6 and 7 of \cite{CP1}; we fix, however, one such decomposition, and denote by $\M^\kappa$ the open quantum walk induced by $\M$ on $\mathcal I_1(\H^\kappa)$.
We define for every $\kappa$ in $K$ the set
$V^\kappa = \{i\in V\, |\, \mathfrak h^\kappa_i \neq\{0\} \}.$

For any $i$ in $V$, $\mathfrak h_i$ has a decomposition $\mathfrak h^{\mathcal D}\oplus\bigoplus_{\kappa\in K} \mathfrak h^\kappa_i$. Then, any state $\rho_i$ on~$\mathfrak h_i$ can be written in block matrix form $\rho=(\rho_i^{\kappa,\kappa'})_{\kappa,\kappa'\in K\cup\{0\}}$, where the index $0$ corresponds to $\mathcal D$. We denote by $\rho_i^{\kappa}$ the diagonal blocks:  $\rho_i^{\kappa} = \rho_i^{(\kappa,\kappa)}$. 
The main tool in this section is a simple observation: for any path $\pi\in \cP_\ell(i,j)$ starting at $i$, the probability of observing, as $\ell$ first steps, the trajectory $\pi=(i_0,\ldots,i_\ell)$ with initial conditions $(i,\rho)$ satisfies
\begin{equation}\label{eq_DecompPathProb}
\mathbb P_{i,\rho}(i_0,\ldots,i_\ell)
\geq \sum_{\kappa\in K} \tr\big((L_\pi \rho L_\pi^*)^{\kappa}\big)
\geq \sum_{\kappa\in K} \tr(L_\pi \rho^{\kappa} L_\pi^*).
\end{equation}
In addition, if $\rho$ has support in $\mathcal R$, then inequalities in \eqref{eq_DecompPathProb} become an equality. If we denote by e.g. $\mfP_{j,i}^{\kappa}$ the operator $\mfP_{j,i}$ associated with the OQW $\M^*$, etc. then we have the following result:
\begin{prop}
	Assume that the open quantum walk $\M$ admits a decomposition \eqref{eq_IrreducibleDecomposition}. Then, for any $\rho$ in $\mathcal S(\H)$,
	\begin{align*}
		\pp_{i,\rho}(t_j<\infty)&\geq \sum_{\kappa\in K} \tr\,\rho^{\kappa}\times \tr\, \mfP_{j,i}\big(\frac{\rho^{\kappa}}{\tr\,\rho^{\kappa}}\big) \, \ind_{i,j \in V^{\kappa}},\\
		\ee_{i,\rho}(n_j)&\geq \sum_{\kappa\in K} \tr\,\rho^{\kappa}\times \tr\, \mfN_{j,i}\big(\frac{\rho^{\kappa}}{\tr\,\rho^{\kappa}}\big) \, \ind_{i,j \in V^{\kappa}},\\
		\ee_{i,\rho}(t_j)&\geq \sum_{\kappa\in K} \tr\,\rho^{\kappa}\times \tr\, \mfT_{j,i}\big(\frac{\rho^{\kappa}}{\tr\,\rho^{\kappa}}\big) \, \ind_{i,j \in V^{\kappa}},\\
		\pp_{i,\rho}(t_\bD)&\geq \sum_{\kappa\in K} \tr\,\rho^{\kappa}\times \tr\, \mfPD_{j,i}\big(\frac{\rho^{\kappa}}{\tr\,\rho^{\kappa}}\big) \, \ind_{i \in V^{\kappa}, \,\bD\cap V^{\kappa}\neq \emptyset}.
	\end{align*}
(where as before we assume that $D$ is a finite domain such that $\bD\neq \emptyset$), and each of these inequalities becomes an equality if we assume that $\rho$ has support in $\mathcal R$.
\end{prop}
If the support of $\rho$ is contained in $\mathcal R$ (in particular if $\mathcal D=\{0\}$) then since $\sum_{\kappa\in K} \tr\,\rho^{\kappa}=1$, it is easy to characterize e.g. the equality $\pp_{i,\rho}(t_j<\infty)=1$, or the finiteness of $\ee_{i,\rho}(n_j)$, in terms of the sub-open quantum walks $\M^{\kappa}$ (to which our various results for irreducible open quantum walks apply).

\section{Variational approach to the Dirichlet problem}
\label{sec_varDirichlet}

We assume throughout this section that $\rhoinv=\sum_{i\in V}{\rhoinv(i)\otimes\proj{i}{i}}$ is an invariant state for the OQW $\M$, which furthermore is faithful. 
In all of this section we write $\tau_\diamond$ instead of $\rhoinv$, i.e. we let
\[\tau_\diamond:=\rhoinv.\]

Our goal in this section is to characterize the solutions of the Dirichlet problem given by Equation \eqref{eq_pbDirichletD} as minimizers of a certain functional, involving the Dirichlet form associated to the OQW. The present Dirichlet forms are simple, discrete-time versions of the non-commutative extensions of classical Dirichlet forms (such extensions were studied first by Davies and Lindsay in \cite{DL1}, see also \cite{DL2,Cip,Cip1}).

We first focus on the definition of the Dirichlet form and its properties. Define on $\B(\H)$ the scalar product
\begin{equation}
\sca{X}{Y}_{\diamond} :=\tr\big(\tau_\diamond^{1/2}X^*\tau_\diamond^{1/2}Y\big).
\label{eq_scal}
\end{equation}
\begin{defi}
The Dirichlet form associated to the open quantum walk $\M$ is the quadratic form
	\begin{equation} 	\label{eq_defiDirichletform}
	\E(X,Y) := \sca{X}{\left(I-\M^*\right)(Y)}_{\diamond}
	\end{equation}
	We also denote $ \E(X)= \E(X,X)$, for $X\in \B(\Hcal)$.
\end{defi}

The central hypothesis in the following is the \emph{detailed balance condition}.
\begin{defi}
	We say that the open quantum walk $\M$ satisfies the detailed balance condition with respect to $\tau_\diamond$ if $\M^*$ is selfadjoint with respect to the scalar product $\langle{\,\cdot\,},{\,\cdot\,}\rangle_\diamond$.
\end{defi}
Note that we reserve the notation $\M^*$ for the adjoint of $\M$ with respect to the duality between $\mathcal I^1(\mathcal H)$ and $\mathcal B(\mathcal H)$. The detailed balance condition is satisfied, in particular, if
\[\tau_\diamond(i)^{1/2} L_{j,i}^*= L_{i,j}\,\tau_\diamond(j)^{1/2} \ \mbox{for all }i,j\in V.\]
It has two immediate consequences:
\begin{lemme} \label{lemme_ptesQDB}
	If the open quantum walk $\M$ satisfies the detailed balance condition, then $\Ecal(X)\geq 0$ for any $X\in\mathcal B(\mathcal H)$. If in addition $\M$ is irreducible, then $\Ecal(X)=0$ if and only if $X\in \cc\id_{\mathcal H}$.
\end{lemme}

\subsection{Dirichlet problem on the whole domain}
Quantum harmonic operators are easily characterized as minimizers of the Dirichlet form. Indeed, the detailed balance condition implies that $\E(X)\geq0$, with equality if and only if $(I-\M^*)(X)=0$, that is, if $X$ is harmonic.
\begin{prop}
	Suppose that $\M$ satisfies the detailed balance condition. Then $X$ is a quantum harmonic observable if and only if $\E(X)$ is minimal, if and only if $\E(X)=0$.
	\label{prop_varharmonic}
\end{prop}

\subsection{Dirichlet problem on a sub-domain}

We now focus on the Dirichlet problem on a finite domain $D\subset V$, that we suppose to be non-empty. Recall the definition of respectively the inner data $A$ and the outer data $B$ as
\begin{equation*}
A= \sum_{i\in D} A_i\otimes \ketbra ii \qquad B=\sum_{j\in \bD} B_j\otimes\ketbra jj,
\end{equation*}
where $A_j,B_j\in\h_j$ for all $j$.\\

Our theorem is the following. An analogue of this result can be found in \cite{Cip1} for more general non-commutative Dirichlet forms.
\begin{theo}
	Let $\M$ be an irreducible open quantum walk with detailed balance condition and $D$ a finite domain of $V$ such that $\bD\neq \emptyset$.
     Then any solution of the Dirichlet problem
	\begin{equation*}
	\left\{
    \begin{array}{r@{\ = \ }l@{\ }l}
        (\id-\M^*)(Z)_i& A_i & \mbox{ for }i\in D \\
        Z_j&B_j & \mbox{ for }j\in \bD.
    \end{array}
    \right.
	\end{equation*}
	is of the form $X_0+B+ Y$, where $Y$ has support in $\mathcal H_{V\setminus(D\cup\bD)}$ and $X_0$ is the unique minimizer over the set $\Bcal\left(\Hcal_D\right)$ of the functional
	\begin{equation}
	E(X)=\frac{1}{2}\,\E(X)+\Ecal(X,B)-\sca{A}{X}_\diamond.
	\label{eq_propVarDirichletpb1}
	\end{equation}
	\label{theo_VarDirichletpb}
\end{theo}

\subsection{The case of doubly stochastic open quantum walks}

In this last section we point out that the Dirichlet form can alternatively be written in terms of first order discrete derivatives (to be defined below) in the special case of doubly stochastic OQW, i.e. for open quantum walks that satisfy $\mathfrak{M}^*=\mathfrak{M}$. 

\begin{prop}\label{prop_doublysto} 
Let $\mathfrak{M}$ be a doubly stochastic open quantum walk satisfying $\mathfrak{M}^*=\mathfrak{M}$. Then $\mathcal{E}(X)$ equals
\begin{equation} \label{eq_first-order}
\frac{1}{2} \sum_{i,j\in V} \mathrm{Tr}\big((\nabla X)_{i,j}(\nabla X)_{i,j}^*\big)=: \frac{1}{2} \|(\nabla X)\|_V^2 
\end{equation}
where $(\nabla X)_{i,j}= X_iL_{i,j}-L_{i,j}X_j$ for $X=\sum_{j\in V} X_j\otimes |j\rangle\langle j| \in \mathcal{B}(\mathcal{H})$.
\end{prop}
Positivity of the Dirichlet form is then manifest in \eqref{eq_first-order}. Notice that the passage from the definition of the Dirichlet form in \eqref{eq_defiDirichletform} to the formula \eqref{eq_first-order} amounts to an integration by part. This presentation of the Dirichlet form in terms of first order difference operators can easily be extended to finite sub-domain if one includes appropriate boundary terms arising from the discrete integration by part.


\bibliography{biblio}

\def\cprime{$'$}
\begin{thebibliography}{10}

\bibitem{Accardirusse}
L.~Accardi.
\newblock {The noncommutative {M}arkov property}.
\newblock {\em Funkcional. Anal. i Prilo\v zen.}, 9(1):1--8, 1975.

\bibitem{AccardiFrigerio}
L.~Accardi and A.~Frigerio.
\newblock {Markovian cocycles}.
\newblock {\em Proc. Roy. Irish Acad. Sect. A}, 83(2):251--263, 1983.

\bibitem{AccardiKoroliuk}
L.~Accardi and D.~Koroliuk.
\newblock {Stopping times for quantum {M}arkov chains}.
\newblock {\em J. Theoret. Probab.}, 5(3):521--535, 1992.

\bibitem{AGS}
S.~Attal, N.~Guillotin-Plantard, and C.~Sabot.
\newblock {Central limit theorems for open quantum random walks and quantum
  measurement records}.
\newblock {\em Ann. Henri Poincar{\'e}}, 16(1):15--43, 2015.

\bibitem{APSS}
S.~Attal, F.~Petruccione, C.~Sabot, and I.~Sinayskiy.
\newblock {Open quantum random walks}.
\newblock {\em J. Stat. Phys.}, 147(4):832--852, 2012.

\bibitem{BBTbistability}
M.~Bauer, D.~Bernard, and A.~Tilloy.
\newblock {Open quantum random walks: Bistability on pure states and
  ballistically induced diffusion}.
\newblock {\em Phys. Rev. A}, 88:062340, Dec 2013.

\bibitem{BBToqbm}
M.~Bauer, D.~Bernard, and A.~Tilloy.
\newblock {The open quantum {B}rownian motions}.
\newblock {\em J. Stat. Mech. Theory Exp.}, 2014(9):p09001, 48, 2014.

\bibitem{Brezis}
H.~Brezis.
\newblock {\em {Functional analysis, {S}obolev spaces and partial differential
  equations}}.
\newblock {Universitext}. Springer, New York, 2011.

\bibitem{CP2}
R.~Carbone and Y.~Pautrat.
\newblock {Homogeneous open quantum random walks on a lattice}.
\newblock {\em J. Stat. Phys.}, 160(5):1125--1153, 2015.

\bibitem{CP1}
R.~Carbone and Y.~Pautrat.
\newblock {Open quantum random walks: reducibility, period, ergodic
  properties}.
\newblock {\em Ann. Henri Poincar{\'e}}, 17(1):99--135, 2016.

\bibitem{CarvalhoGuidiLardizabal}
S.~L. {Carvalho}, L.~F. {Guidi}, and C.~F. {Lardizabal}.
\newblock {Site recurrence of open and unitary quantum walks on the line}.
\newblock {\em ArXiv e-prints}, July 2016.

\bibitem{Cip1}
F.~Cipriani.
\newblock {The variational approach to the Dirichlet problem in
  $C^*$-algebras}.
\newblock {\em Banach center publications}, 43, 1998.

\bibitem{Cip}
F.~Cipriani.
\newblock {Dirichlet forms on noncommutative spaces}.
\newblock In {\em {Quantum potential theory}}, pages 161--276. Springer, 2008.

\bibitem{Dav}
E.~B. Davies.
\newblock {Quantum stochastic processes. {II}}.
\newblock {\em Comm. Math. Phys.}, 19:83--105, 1970.

\bibitem{DL1}
E.~B. Davies and J.~M. Lindsay.
\newblock {Non-commutative symmetric Markov semigroups}.
\newblock {\em Mathematische Zeitschrift}, 210(1):379--411, 1992.

\bibitem{DL2}
E.~B. Davies and J.~M. Lindsay.
\newblock {Superderivations and symmetric Markov semigroups}.
\newblock {\em Communications in Mathematical Physics}, 157(2):359--370, 1993.

\bibitem{DhahriMukhamedov}
A.~{Dhahri} and F.~{Mukhamedov}.
\newblock {Open Quantum Random Walks and associated Quantum Markov Chains}.
\newblock {\em ArXiv e-prints}, Aug. 2016.

\bibitem{Durrett}
R.~Durrett.
\newblock {\em {Probability: theory and examples}}.
\newblock {Cambridge Series in Statistical and Probabilistic Mathematics}.
  Cambridge University Press, Cambridge, fourth edition, 2010.

\bibitem{EHK}
D.~E. Evans and R.~H{\o}egh-Krohn.
\newblock {Spectral properties of positive maps on {$C\sp*$}-algebras}.
\newblock {\em J. London Math. Soc. (2)}, 17(2):345--355, 1978.

\bibitem{FagnolaRebolledo2003}
F.~Fagnola and R.~Rebolledo.
\newblock {Transience and recurrence of quantum {M}arkov semigroups}.
\newblock {\em Probab. Theory Related Fields}, 126(2):289--306, 2003.

\bibitem{FNW}
M.~Fannes, B.~Nachtergaele, and R.~F. Werner.
\newblock {Finitely correlated states on quantum spin chains}.
\newblock {\em Comm. Math. Phys.}, 144(3):443--490, 1992.

\bibitem{Gro}
U.~Groh.
\newblock {The peripheral point spectrum of {S}chwarz operators on {$C^{\ast}
  $}-algebras}.
\newblock {\em Math. Z.}, 176(3):311--318, 1981.

\bibitem{GVWW}
F.~A. Gr{\"u}nbaum, L.~Vel{\'a}zquez, A.~H. Werner, and R.~F. Werner.
\newblock {Recurrence for discrete time unitary evolutions}.
\newblock {\em Comm. Math. Phys.}, 320(2):543--569, 2013.

\bibitem{Kato}
T.~Kato.
\newblock {\em {Perturbation Theory for Linear Operators}}.
\newblock {Classics in mathematics}. Springer, 1976.

\bibitem{KY}
N.~Konno and H.~Yoo.
\newblock {Limit Theorems for Open Quantum Random Walks}.
\newblock {\em Journal of Statistical Physics}, 150(2):299--319, 2013.

\bibitem{KuMa}
B.~K{\"u}mmerer and H.~Maassen.
\newblock {A pathwise ergodic theorem for quantum trajectories}.
\newblock {\em J. Phys. A}, 37(49):11889--11896, 2004.

\bibitem{LardSouza}
C.~F. {Lardizabal} and R.~R. {Souza}.
\newblock {On a class of quantum channels, open random walks and recurrence}.
\newblock {\em ArXiv e-prints}, Feb. 2014.

\bibitem{LardSouza2}
C.~F. {Lardizabal} and R.~R. {Souza}.
\newblock {Open Quantum Random Walks: Ergodicity, Hitting Times, Gambler's Ruin
  and Potential Theory}.
\newblock {\em Journal of Statistical Physics}, July 2016.

\bibitem{LawlerLimic}
G.~F. Lawler and V.~Limic.
\newblock {\em {Random walk: a modern introduction}}, volume 123 of {\em
  {Cambridge Studies in Advanced Mathematics}}.
\newblock Cambridge University Press, Cambridge, 2010.

\bibitem{Lim2010}
B.~J. Lim.
\newblock {\em {Fronti{\`e}res de Poisson d'op{\'e}ration quantiques et
  trajectoires quantiques}}.
\newblock PhD thesis, 2010.
\newblock Th{\`e}se de doctorat dirig{\'e}e par Bekka, Bachir et Petritis,
  Dimitri Math{\'e}matiques et applications Rennes 1 2010.

\bibitem{Norris}
J.~R. Norris.
\newblock {\em {Markov chains}}, volume~2 of {\em {Cambridge Series in
  Statistical and Probabilistic Mathematics}}.
\newblock Cambridge University Press, Cambridge, 1998.
\newblock Reprint of 1997 original.

\bibitem{PelOQRW}
C.~Pellegrini.
\newblock {Continuous time open quantum random walks and non-{M}arkovian
  {L}indblad master equations}.
\newblock {\em J. Stat. Phys.}, 154(3):838--865, 2014.

\bibitem{Revuz}
D.~Revuz.
\newblock {\em {Markov chains}}, volume~11 of {\em {North-Holland Mathematical
  Library}}.
\newblock North-Holland Publishing Co., Amsterdam, second edition, 1984.

\bibitem{RussoDye}
B.~Russo and H.~A. Dye.
\newblock {A note on unitary operators in {$C^{\ast} $}-algebras}.
\newblock {\em Duke Math. J.}, 33:413--416, 1966.

\bibitem{Sch}
R.~Schrader.
\newblock {Perron-{F}robenius theory for positive maps on trace ideals}.
\newblock In {\em {Mathematical physics in mathematics and physics ({S}iena,
  2000)}}, volume~30 of {\em {Fields Inst. Commun.}}, pages 361--378. Amer.
  Math. Soc., Providence, RI, 2001.

\bibitem{SinPet}
I.~Sinayskiy and F.~Petruccione.
\newblock {Open Quantum Walks: a short introduction}.
\newblock {\em Journal of Physics: Conference Series}, 442(1):012003, 2013.

\bibitem{wolftour}
M.~M. Wolf.
\newblock {Quantum channels \& operations: Guided tour}.
\newblock
  \url{http://www-m5.ma.tum.de/foswiki/pub/M5/Allgemeines/MichaelWolf/QChannelLecture.pdf},
  2012.
\newblock Lecture notes based on a course given at the Niels-Bohr Institute.

\bibitem{ZWLJN}
J.~{Zhang}, Y.-x. {Liu}, R.-B. {Wu}, K.~{Jacobs}, and F.~{Nori}.
\newblock {Quantum feedback: theory, experiments, and applications}.
\newblock {\em ArXiv e-prints}, July 2014.

\end{thebibliography}

\appendix

\section{Proofs for Section \ref{sec_definitionsnotations}}
\paragraph{Proof of Lemma \ref{lemme_martingale}:}
Conditionally on $(x_n,\rho_n)$, one has for all $i$ in $V$
\[m_{n+1}=\tr(\rho_{n+1} A_{x_{n+1}}) = \tr\big(\frac{L_{i,x_n}\rho_{n}L_{i,x_n}^*}{\tr(L_{i,x_n}\rho_{n}L_{i,x_n}^*)} A_{i}\big)\]
with probability $\tr(L_{i,x_n}\rho_{n}L_{i,x_n}^*)$, so that
\begin{align*}
\ee\big(\tr(\rho_{n+1} A_{x_{n+1}})| x_n,\rho_n\big) &= \sum_{i\in  V}  \tr\big({L_{i,x_n}\rho_{n}L_{i,x_n}^*} A_{i}\big)\\
&= \tr\big(\rho_{n} \sum_{i\in V} L_{i,x_n}^* A_{i} L_{i,x_n}\big)\\
&= \tr (\rho_n A_{x_n})=m_n.
\end{align*} \qed

\section{Proofs for Section \ref{sec_passagetimesnumbervisits}} \label{sec_appproofspassagetimes}
We start by computing simple expressions for the quantities $\pp_{i,\rho}(t_j<\infty)$ and $\ee_{i,\rho}(n_j) $:
\begin{lemme} \label{lemme_identities}
    We have the identities
    \begin{gather*}
    \pp_{i,\rho}(t_j<\infty) = \hspace{-0.6em} \sum_{\pi\in \cP^{V\setminus\{j\}}(i,j)}\hspace{-0.6em} \tr (L_\pi \rho L_\pi^*),\qquad
    \ee_{i,\rho}(n_j) = \sum_{\pi\in \cP(i,j)} \tr (L_\pi \rho L_\pi^*),
    \end{gather*}
where the second expression is possibly $\infty$.
\end{lemme}
\begin{proof}
We have $\pp_{i,\rho}(x_1=i_1,\ldots,x_\ell=i_\ell)=\tr\, L_\pi \rho L_\pi^*$ where $\pi=(i,i_1,\ldots,i_\ell)$. In addition,
\[\pp_{i,\rho}(t_j<\infty)= \sum_{k\geq 0}\ \sum_{i_1,\ldots,i_k\in V\setminus\{j\}} \pp_{i,\rho}(x_1=i_1,\ldots,x_k=i_k,x_{k+1}=j)\]
which leads to the first formula. We also have immediately $\pp_{i,\rho}(x_k=j)=\sum_{\pi\in \cP_k(i,j)} \tr(L_\pi \rho L_\pi^*)$, and the second formula follows from
\[\ee_{i,\rho}(n_j)=\sum_{k= 1}^\infty\pp_{i,\rho}(x_k=j).\]
\end{proof}

\paragraph{Proof of Proposition \ref{prop_ReptjT}}

We begin with the definition of $\mfP_{i,j}$. For any $\rho$ in $\mathcal I_1(\h_i)\setminus\{0\},$ the triangle inequality for the trace norm 
implies that
\begin{align*}
    \tr\big(\big|\sum_{\pi\in\mathcal P^{V\setminus\{j\}}(i,j)} L_\pi \rho L_\pi^*\, \ind_{\ell(\pi)\leq n}\big|\big)&\leq \sum_{\pi\in\mathcal P^{V\setminus\{j\}}(i,j)} \tr\, \big|L_\pi \rho L_\pi^*\big|\, \ind_{\ell(\pi)\leq n}\\
    &= \sum_{\pi\in\mathcal P^{V\setminus\{j\}}(i,j)} \tr\, (L_\pi |\rho| L_\pi^*)\, \ind_{\ell(\pi)\leq n}\\
    &= \tr|\rho| \times \pp_{i,\frac{|\rho|}{\tr(|\rho|)}}(t_j\leq n) \\
    &\leq \tr|\rho|,
\end{align*}
so that 
\[\sup_n \tr\big(\big|\sum_{\pi\in\mathcal P^{V\setminus\{j\}}(i,j)} L_\pi \rho L_\pi^* \ind_{\ell(\pi)\leq n}\big|\big) <\infty.\]
Consequently, by the Banach-Steinhaus Theorem, the operator on $\mathcal I_1(\h_i)$ defined by 
\[\mfP_{j,i}(\rho)= \lim_{n\to\infty}\quad  \sum_{\pi\in\mathcal P^{V\setminus\{j\}}(i,j)} L_\pi \rho L_\pi^* \ind_{\ell(\pi)\leq n}\]
is everywhere defined and bounded.

This proves the first identity in Proposition \ref{prop_ReptjT}. To prove the second we need a series of technical results. Our strategy is the same as in the classical case: we introduce a weight on the length of paths, in order to tame the possible divergence of the series giving $\ee_{i,\rho}(n_j)$ in Lemma \ref{lemme_identities}. First note that, for any $i,j\in V$ and any $\alpha\in (0,1)$, there exists a bounded, completely positive map $\mfN_{j,i}^{(\alpha)}$ from $\mathcal I_1(\h_i)$ to $\mathcal I_1(\h_j)$ such that 
\[\sum_{\pi \in \mathcal P(i,j)}\alpha^{\ell(\pi)} \tr \,L_\pi \rho L_\pi^* = \tr \,\mfN_{j,i}^{(\alpha)}(\rho).\]
In particular, the following limit holds in $[0,\infty]$:
\[\ee_{i,\rho}(n_j)=\lim_{\alpha\to 1}\tr \, \mfN_{j,i}^{(\alpha)}(\rho).\]
This operator $\mfN_{j,i}^{(\alpha)}$ is defined by
\[\mfN_{j,i}^{(\alpha)}(\rho)=\lim_{n\to\infty} \sum_{\pi\in \cP(i,j)} \alpha^{\ell(\pi)} L_\pi \rho L_\pi^* \ind_{\ell(\pi)\leq n},\]
using the Banach-Steinhaus Theorem and the simple bound
\begin{align*}
    \tr \big(\big|\hspace{-0.5em}\sum_{\pi\in \cP(i,j)} \alpha^{\ell(\pi)}\, L_\pi \rho L_\pi^* \ind_{\ell(\pi)\leq n} \big|\big)
    &\leq \sum_{\pi\in\cP(i,j)} \alpha^{\ell(\pi)}\, \tr \,L_\pi |\rho| L_\pi^* \, \ind_{\ell(\pi)\leq n}\\
    &= \sum_{k= 0}^n \alpha^k \,\pp_{i,\rho} (x_k=j)\\
    &\leq (1-\alpha)^{-1}.
    \end{align*}
We also define
\[\mfP_{j,i}^{(\alpha)}(\rho)= \lim_{n\to\infty}\,  \sum_{\pi\in\mathcal P^{V\setminus\{j\}}(i,j)} \alpha^{\ell(\pi)}\,L_\pi \rho L_\pi^* \,\ind_{\ell(\pi)\leq n}.\]
Since any $\pi\in \cP(i,j)$ is a concatenation of $\pi_0 \in \cP^{V\setminus\{j\}}(i,j)$ and $\pi_1,\ldots,\pi_k$ in $\cP^{V\setminus\{j\}}(j,j)$, and
\[L_\pi = L_{\pi_k}\circ \ldots \circ L_{\pi_1} \circ L_{\pi_0},\qquad \ell(\pi)= \ell(\pi_k)+\ldots + \ell(\pi_1)+\ell(\pi_0),\]
we have
\begin{align*}
     &\sum_{\pi\in \cP(i,j) }\alpha^{\ell(\pi)} L_\pi \rho  L_\pi^* \ind_{\ell(\pi)\leq n}\\
     &=\sum_{k\geq 0}\quad \sum_{\substack{\pi_0\in\cP^{V\setminus\{j\}}(i,j),\\\pi_1,\ldots,\pi_k\in\cP^{V\setminus\{j\}}(j,j)}} \alpha^{\sum_{r=0}^k \ell(\pi_r)}\, L_{\pi_k} \ldots  L_{\pi_1}  L_{\pi_0} \rho L_{\pi_0}^* L_{\pi_1}^*  \ldots   L_{\pi_k}^*\, \ind_{\sum_{r=0}^k \ell(\pi_r)\leq n}. 
 \end{align*} 
Because both sides define bounded operators as $n\to\infty$, we have
\[\mfN_{j,i}^{(\alpha)}(\rho)=\sum_{k\geq 0} \mfP_{j,j}^{(\alpha)\,k} \circ \mfP_{i,j}(\rho)= (\id-\mfP_{j,j}^{(\alpha)})^{-1} \circ \mfP_{j,i}^{(\alpha)}(\rho).\]
Since $\alpha \mapsto \mfP_{j,j}^{(\alpha)}(\rho)$ is monotone increasing for $\rho \geq 0$, the right-hand side is monotone increasing as well, and the second identity follows. 

\qed

\paragraph{Proof of Equation \eqref{eq_esprhoti}}
By definition, we have
\begin{align*}
\ee_{i,\rho}(\rho_{t_j}\,|\, t_j<\infty) &= \frac{\ee_{i,\rho}(\rho_{t_j} \ind_{t_j<\infty})}{\pp_{i,\rho}(t_j<\infty)}\\
&=  \frac1{\pp_{i,\rho}(t_j<\infty)} \, \sum_{\pi\in \cP^{V\setminus\{j\}}(i,j)} \frac{L_\pi \rho L_\pi^*}{\tr\,L_\pi \rho L_\pi^*} \, \tr\,L_\pi \rho L_\pi^*\\
&= \frac{\mfP_{j,i}(\rho)}{\tr \,\mfP_{j,i}(\rho)}.
\end{align*}
\qed

\paragraph{Proof of Corollary \ref{coro_ReptjT}}\,\hfill
\begin{enumerate}
\item Let $i,j\in V$ and $\rho\in\Scal(\h_i)$. By Proposition \ref{prop_ReptjT}, we have $\pp_{i,\rho}(t_j<\infty)= \tr\, \rho\,\mfP^*_{j,i}(\id_{\h_{j}})$ and, since $\tr\,\rho=1$, we have $\pp_{i,\rho}(t_j<\infty)~=~1$ if and only if $P_\rho \mfP_{j,i}^*(\id_{\h_i}) P_\rho = P_\rho$, where $P_\rho$ is the orthogonal projection on the support of $\rho$. Write $\mfP_{j,i}^*(\id_{\h_j})$ as
$\mfP_{j,i}^*(\id_{\h_j})=\begin{pmatrix} \id_{\mathrm{Ran}\,\rho} & A \\ A^* & B\end{pmatrix}$ in the decomposition $\h_i=\mathrm{Ran}\,\rho \oplus (\mathrm{Ran}\,\rho)^\perp$. Then the property $\mfP_{j,i}^*(\id_{\h_j})\leq \id_{\h_i}$ implies that $\begin{pmatrix} 0 & -A \\ -A^* & \id_{\mathrm{Ker}\,\rho} - B \end{pmatrix}\geq 0$, so that necessarily $A=0$. In particular, if $\rho$ is faithful, then $\pp_{i,\rho}(t_j<\infty)=1$ if and only if $\mfP_{j,i}^*(\id_{\h_j})=\id_{\h_i}$. In that case, $\pp_{i,\rho'}(t_j<\infty)=1$ for any $\rho'$ in $\Scal(\h_i)$.
\item Consequently, if this is the case for $j=i$, then for any $\rho'$ in $\Scal(\h_i)$ one has $\ee_{i,\rho'}(n_i)=\infty$, since by Proposition \ref{prop_ReptjT} we have
\[\ee_{i,\rho'}(n_i)=\sum_{k\geq 1} \tr\, \rho'\, \mfP_{i,i}^{*\, k}(\id_{\h_i}).\] 
\item If $\ee_{i,\rho}(n_j)<\infty$ with $\rho$ faithful and $\mathrm{dim}\,\h_i<\infty$, then for any $\alpha\in(0,1)$,
\[\tr\,\mfN_{j,i}(\rho)\geq \tr\big(\rho\mfN_{j,i}^{(\alpha)\, *}(\id_{\h_j})\big) \geq \inf\,\sp (\rho) \times \|\mfN_{j,i}^{(\alpha)\, *}(\id_{\h_i})\|,\]
so that $\mfN_{j,i}^{(\alpha)\, *}(\id_{\h_j})$ is uniformly (in $\alpha$) bounded in norm. The monotone increasing function $\alpha\mapsto \mfN_{j,i}^{(\alpha)\, *}(\id_{\h_j})$ therefore has a limit and, by Proposition \ref{prop_ReptjT}, $\ee_{i,\rho'}(n_j)<\infty$ for any $\rho'$.
\item The construction of $\mfN_{j,i}$ when $\ee_{i,\rho}(n_j)<\infty$ for any $\rho$ is obtained by a Banach-Steinhaus argument.\qed
\end{enumerate}
\paragraph{Proof of Proposition \ref{prop_Markovsimple}}
	Recall that $\pp_{i,\rho}(t_i<\infty)=\|\mfP_{i,i}(\rho)\|$. By Proposition \ref{prop_ReptjT}, the map $\mfP_{i,i}$ is bounded, and since $\mathcal S(\h_i)$ is compact, the supremum $p=\sup_{\rho\in \mathcal S(\h_i)} \tr\,\mfP_{i,i}(\rho)$ satisfies $p<1$. A standard application of the strong Markov property for the chain $(x_n,\rho_n)_n$ shows that $\pp_{i,\rho}(n_i=k)\leq p^k$ and by a direct computation $\ee_{i,\rho}(n_i) \leq p(1-p)^{-2}$, which gives the result.
\qed 

\paragraph{Proof of Proposition \ref{prop_ExpectedNumberVisits} and Corollary \ref{coro_ExpectedNumberVisits}}
We start with two simple lemmata:
\begin{lemme} \label{lemme_probaminoree}
Assume that $\M$ is an irreducible open quantum walk and let $i,j$ in $V$ be such that $\mathrm{dim}\,\h_i<\infty$. Then  
\[\inf_{\rho\in\mathcal S(\h_i)}\pp_{i,\rho}(t_j<\infty)>0.\]
\end{lemme}
\begin{proof}[Proof of Lemma \ref{lemme_probaminoree}]
    For any $\rho$ in $\mathcal S(\h_i)$, there exists a unit vector $\varphi$ in $\h_i$ and $\lambda > 0$ such that $\rho\geq \lambda\ketbra\varphi\varphi$. By irreducibility, there exists a path $\pi$ in $\mathcal P(i,j)$ such that $\|L_\pi\varphi\|^2>0$, so that $\pp_{i,\rho}(t_{j}<\infty)>0$. By continuity of $\mfP_{j,i}$ and compactness of $\mathcal S(\h_j)$, one has the result.
\end{proof}

\begin{lemme} \label{lemme_universality}
Assume that $\M$ is an irreducible open quantum walk and let $i,j$ be in $V$. If $\mathrm{dim}\,\h_j<\infty$ and $\rho\in\mathcal S(\h_i)$ is such that $\ee_{i,\rho}(n_j)=\infty$, then for any $j' \in V$ one has $\ee_{i,\rho}(n_{j'})=\infty$.
\end{lemme}

\begin{proof}[Proof of Lemma \ref{lemme_universality}]
    By Lemma \ref{lemme_probaminoree}, one has $\inf_{\rho'\in\mathcal S(\h_j)}\pp_{j,\rho'}(t_{j'}<\infty)>0$ for any $j'\in V$. Now, a standard markovianity argument shows that $\ee_{i,\rho}(n_j)=\infty$ implies $\ee_{i,\rho}(n_{j'})=\infty$.
\end{proof} 

\begin{remark}
    Here we used only a weaker version of irreducibility, namely the fact that for any $k,l$ in $V$, any $\varphi$ in $\h_k$, there exists a path $\pi$ in $\cP(k,l)$ such that $L_\pi \varphi\neq 0$.
\end{remark}

Let us go back to the proof of Proposition \ref{prop_ExpectedNumberVisits} and Corollary \ref{coro_ExpectedNumberVisits}. Define for $j$ in $V$
\begin{equation} \label{eq_defDUA}
  D^n(j)=\big\{\varphi=\sum_{i\in V}\varphi_i\otimes \ket i \ \mbox{ s.t. } \sum_{i\in V}\sum_{\pi\in\cP(i,j)} \|L_\pi\varphi_i\|^2 < \infty\big\}.
\end{equation}
It is immediate that $D^n(j)$ is a vector space, and that $(L_{k,l}\otimes\ketbra kl) D^n(j)\subset D^n(j)$ for any $k,l$ in $V$. In the language of \cite{CP1}, this means that $\overline{D^n (j)}$ is an enclosure for $\M$. Moreover, the only possible enclosures for an irreducible $\M$ are $\{0\}$ and $\H$. Therefore, either ${D^n(j)}=\{0\}$ or $\overline{D^n(j)}=\H$. Define for $i$ in $V$ ${\mathfrak d}^n _{j,i}=D^n(j)\cap \h_i$ (with a slight abuse of notation). Then either for every $i$ the subspace ${\mathfrak d}^n_{j,i}$ is dense in $\h_i$ or for every $i$ it is $\{0\}$.
Remark that by Lemma \ref{lemme_identities}, $\sum_{\pi\in\cP(i,j)} \|L_\pi\varphi_i\|^2 = \ee_{i,\ketbra{\varphi_i}{\varphi_i}}(n_j)$. 
    By linearity of $\ee_{i,\rho}(n_j)$ in $\rho$, if ${\mathfrak d}^n_{j,i}=\{0\}$ then $\ee_{i,\rho}(n_j)=\infty$ for any $\rho$ in $\mathcal S(\h_i)$, and if ${\mathfrak d}^n_{j,i}$ is dense then $\ee_{i,\rho}(n_j)<\infty$ for any $\rho$ with finite range in ${\mathfrak d}^n_{j,i}$.
This concludes the proof of Proposition \ref{prop_ExpectedNumberVisits}.\\

    Now, if $\mathrm{dim}\,\h_i<\infty$, then in situation \textit{2.} of Proposition \ref{prop_ExpectedNumberVisits} one has $\ee_{i,\rho}(n_j)=\infty$ for any $i$ in $V$ and $\rho$ in $\mathcal S(\h_i)$. Now, Lemma \ref{lemme_universality} forbids the situation where for $j\neq j'$ one has $\ee_{i,\rho}(n_j)=\infty$ and $\ee_{i,\rho}(n_{j'})<\infty$ for every~$\rho$ in $\mathcal S(\h_i)$, and this proves Corollary \ref{coro_ExpectedNumberVisits}. \qed
\begin{remark}
    This proof is essentially due to  \cite{FagnolaRebolledo2003}.
\end{remark}

\paragraph{Proof of Proposition \ref{prop_EnPt}}
	By Definition~\ref{defi_irreducibility} of irreducibility, there is no nontrivial invariant subspace of $\h_j$ left invariant by all $L_\pi$, $\pi\in\cP(j,j)$. Since any $\pi\in\cP(j,j)$ is a concatenation of paths in $\cP^{V\setminus\{j\}}(j,j)$, there is also no nontrivial invariant subspace of $\h_j$ left invariant by all $L_\pi$, $\pi\in\cP^{V\setminus\{j\}}(j,j)$, and this means that $\mfP_{j,j}$ is a completely positive irreducible map on $\mathcal I_1(\h_j)$. In addition, we know from the Russo-Dye Theorem that $\|\mfP_{j,j}\|=\|\mfP^*_{j,j}(\id)\|\leq1$, so that the eigenvalue $\lambda$ of $\mfP_{j,j}$ of largest modulus satisfies $|\lambda|\leq 1$. By the Perron-Frobenius Theorem for completely positive maps acting on the set of trace-class operators of a finite-dimensional space (see Theorem 3.1 and Remark 3.1 in \cite{Sch}, which are essentially proven in \cite{EHK}), there exists a faithful state $\rho_{\mathrm{f}}$ on $\h_j$ such that $\mfP_{j,j}(\rho_{\mathrm{f}})=|\lambda|\rho_{\mathrm{f}}$. If $|\lambda|<1$, then by Proposition \ref{prop_ReptjT} one has $\ee_{j,\rho_{\mathrm{f}}}(n_j)<\infty$. However, by Proposition \ref{prop_ExpectedNumberVisits}, the assumption $\ee_{i,\rho}(n_j)=\infty$ implies $\ee_{j,\rho_{\mathrm{f}}}(n_j)=\infty$, a contradiction. Therefore $|\lambda|=1$, $\rho_{\mathrm{f}}$ is a faithful invariant state and $\tr\, \mfP_{j,j}(\rho_{\mathrm{f}})=\tr\,\rho_{\mathrm{f}}~=~1.$	By Corollary \ref{coro_ReptjT}, we have that $\pp_{j,\rho}(t_j<\infty)=1$ for any $\rho$ in $\mathcal S(\h_i)$. \qed

\section{Proofs for Section \ref{sec_expreturntimes}}

\paragraph{Proof of Proposition \ref{prop_expetj}}
The expansion of $\ee_{i,\rho}(t_j)$ and the construction of~$\mfT_{j,i}$ are obtained by now standard Banach-Steinhaus arguments. \qed

\paragraph{Proof of Proposition \ref{prop_ExpectedTime}} 
Proposition \ref{prop_ExpectedTime} is proved like Proposition \ref{prop_ExpectedNumberVisits}, by introducing 
\begin{equation} \label{eq_defDt}
D^t(j)=\big\{\varphi=\sum_{i\in V}\varphi_i\otimes \ket i \ \mbox{ s.t. } \sum_{i\in V}\sum_{\pi\in\cP^{V\setminus\{j\}}(i,j)} \ell(\pi) \,\|L_\pi\varphi_i\|^2 < \infty\big\}
\end{equation}
and remarking that $D^t(j)$ is an enclosure.\qed

\paragraph{Proof of Theorem \ref{theo_espti}} 
Define $\mathfrak d^t_{j,i}=D^t(j)\cap\h_i$. Remark that in the case of a semifinite OQW, by Proposition \ref{prop_ExpectedTime}, for every $j$ in $V$ either $\mathfrak d^t_{j,i}=\{0\}$  for every $i$; or $\mathfrak d^t_{j,i}=\h_i$ for every $i$. If for some $j$ one has $\mathfrak d^t_{j,i}=\h_i$ for every $i$, then we have in particular $\ee_{j,\rho}(t_{j})<\infty$ for any $\rho$ in $\mathcal S(\h_j)$; for any $j'$, applying Lemma \ref{lemme_probaminoree} again one has  $\inf_{\rho\in\mathcal S(\h_j)}\pp_{j,\rho}(t_{j'}<\infty)>0$. By a markovianity argument, one obtains that $\ee_{j,\rho'}(t_{j'})<\infty$ for any $j'$ in $V$ and $\rho'\in\mathcal S(\h_j)$. \qed

\paragraph{Proof of Theorem \ref{theo_espti2}} 
Let $\rhoinv=\sum_{i\in V} \rhoinv(i)\otimes\ketbra ii$ be an invariant state for $\M$. Then by the infinite-dimensional extension of the K\"ummerer-Maassen ergodic Theorem (see \cite{Lim2010}), one has, for any $i\in V$ and $\rho\in\mathcal S(\h_i)$, the  $\pp_{i,\rho}$-almost-sure convergence
\begin{equation} \label{eq_KuMapveespti2}
\frac1n \sum_{k=1}^{n} \rho_k \otimes \ketbra{x_k}{x_k}\underset{n\to\infty}{\longrightarrow} \sum_{j\in V}\rhoinv(j)\otimes\ketbra jj,
\end{equation}
where convergence is in the weak-* sense. This implies in particular that
\[n_j^{(k)} = \mathrm{card}\{n\leq k\,|\, x_n=j\}\]
satisfies, for any $j\in V$, ${n_j^{(k)}}/k\underset{k\to\infty}\to \tr\,\rhoinv(j)$, $\pp_{i,\rho}$-almost-surely. Therefore, $t_j^{(k)}<\infty$ but $t_j^{(k)}\underset{k\to\infty}\to\infty$. Considering $m=t_j^{(k)}$, we have ${n_j^{(m)}}/{m}={k}/{t_j^{(k)}}$ and therefore, $\pp_{i,\rho}$-almost-surely, $t_j^{(k)}/k \to \big(\tr\,\rhoinv(j)\big)^{-1}$ .

Observe now that, as shown in Example \ref{ex_invariantstate}, our assumptions imply in particular that $\pp_{j,\rho}(t_j<\infty)=1$ for any $\rho$ in $\Scal(\h_j)$, so that $\mfP_{j,j}$ is a completely positive, trace-preserving map, with Kraus decomposition
\[\mfP_{j,j}(\rho)=\sum_{\pi\in\mathcal P^{V\setminus\{j\}}} L_\pi \rho L_\pi^*.\] 
In addition, we have $\pp_{j,\rho}$-almost-surely from \eqref{eq_KuMapveespti2}	
\[\frac1n \sum_k\rho_{t_{j}^{(k)}} \ind_{t_j^{(k)}\leq n} \underset{n\to\infty}{\longrightarrow} \rhoinv(j)\]
(the convergence needs not be specified, as $\h_j$ is finite-dimensional), but the K\"ummerer-Maassen ergodic Theorem applied to $\mfP_{j,j}$ shows that $\frac1{n_j^{(m)}} \sum_{k=1}^{n_j^{(m)}} \rho_{t_{i}^{(k)}}$ converges almost-surely to an invariant of $\mfP_{j,j}$. Therefore, $\frac{\rhoinv(j)}{\tr\,\rhoinv(j)}$ is an invariant state for $\mfP_{j,j}$ and $\pp_{j,\rho}$-almost-surely,
\begin{equation} \label{eq_acontredire}
 	\frac1{n_j^{(m)}} \sum_{k=1}^{n_j^{(m)}} \rho_{t_{i}^{(k)}} \underset{m\to\infty}\longrightarrow \frac{\rhoinv(j)}{\tr\,\rhoinv(j)}
\end{equation}
In addition, since ${\rhoinv(j)}$ is faithful on $\h_j$, one has by necessity that $\mfP_{j,j}$ is irreducible: if there existed an invariant subspace for all $L_\pi$, $\pi\in \mathcal P^{V\setminus\{j\}}$, then there would exist an invariant state $\rho'_j$ for $\mfP_{j,j}$ with support on this invariant subspace, and considering initial data $(j,\rho_j')$ in \eqref{eq_acontredire} above would show that $\rhoinv(j)$ has support no larger than the support of $\rho'_j$, a contradiction.

We now define a new probability space by $\Omega^{(j)}= \big(\mathcal P^{V\setminus\{j\}}(j,j)\big)^{\otimes \nn}$, and let
\[\pp^{(j)}(\pi_1,\ldots,\pi_m)= \tr \big(L_{\pi_m}\ldots L_{\pi_1} \frac{\rhoinv(j)}{\tr\,\rhoinv(j)} L_{\pi_1}^* \ldots L_{\pi_m}^*\big).\]
The trace-preserving property of $\mfP_{j,j}$ shows that this defines a consistent family and by the Daniell-Kolmogorov extension Theorem this defines a probability $\pp^{(j)}$ on $\Omega^{(j)}$. In addition, the invariance of $\frac{\rhoinv(j)}{\tr\,\rhoinv(j)}$ by $\mfP_{j,j}$ implies 
\[\sum_{\pi_1 \in \mathcal P^{V\setminus\{j\}}(j,j)} \pp^{(j)}(\pi_1,\ldots,\pi_m)= \pp^{(j)}(\pi_2,\ldots,\pi_m),\]
which shows that $\pp^{(j)}$  is invariant by the left shift
\begin{equation*}
	\begin{array}{cccc}
		\Theta : & \Omega^{(j)} & \rightarrow & \Omega^{(j)}\\
				 & (\pi_1,\pi_2,\ldots)&\mapsto & (\pi_2,\pi_3,\ldots)
	\end{array}
 \end{equation*} 
Now, the Perron-Frobenius Theorem implies that $1$ is a simple eigenvalue for~$\mfP_{j,j}$. This immediately shows that for any two cylinder sets $E$ and $F$, 
\[\frac1m \sum_{k=1}^m \pp^{(j)}\big(E \cap \Theta^{-k}(F)\big)\underset{m\to\infty} \longrightarrow \pp^{(j)}(E) \,\pp^{(j)}(F),\]
so that $(\Omega^{(j)},\pp^{(j)})$ is ergodic for $\Theta$. Now, if we consider the map $\ell^{(k)}$ defined by
\[\ell^{(k)}(\pi_1,\pi_2,\ldots)=\ell(\pi_1)+\ldots+\ell(\pi_k),\]
then this map satisfies $\ell^{(k+k')}=\ell^{(k)}+\ell^{(k')}\circ \Theta^k$. By Birkhoff's ergodic Theorem one has $\pp^{(j)}$-almost-sure convergence of ${\ell^{(k)}}/k$ to the expectation of $\ell^{(1)}$ for $\pp^{(j)}$. It is immediate, however, that the distribution of $\ell^{(k)}$ under $\pp^{(j)}$ is the same as the distribution of $t_j^{(k)}$ under $\pp_{j,\frac{\rhoinv(j)}{\tr\,\rhoinv(j)}}$. We therefore have
\[{t_j^{(k)}}/{k}\underset{k\to\infty}\longrightarrow \ee_{j,\frac{\rhoinv(j)}{\tr\,\rhoinv(j)}}(t_j^{(1)}),\]
where convergence is both almost-sure and in the $\mathrm L^1$ sense, with respect to $\pp_{j,\frac{\rhoinv(j)}{\tr\,\rhoinv(j)}}$. The first part of the proof shows that 
\[\ee_{j,\frac{\rhoinv(j)}{\tr\,\rhoinv(j)}}(t_j^{(1)})=\big(\tr\,\rhoinv(j)\big)^{-1},\]
and this concludes the proof.
\qed

\section{Proofs for Section \ref{sec_exittimesDirichletproblems}}
\paragraph{Proof of Proposition \ref{prop_probDirichletV}} 

Consider $A=\sum_{i\in V} A_i\otimes \ketbra ii$ such that for any~$i$ in $V$, $\sum_{j\in V}\|\mfN_{j,i}^*(A_j)\|<\infty$. Then \eqref{eq_solutionpbDirichletV} defines an operator $Z$. Proving that $Z$ satisfies \eqref{eq_pbDirichletV} is then a straightforward computation. By linearity it is enough to assume that $A=A_k\otimes \ketbra kk$. We then have
    \begin{align*}
        \M^*(Z)
        &= \sum_{i\in V} \Big(\sum_{j\in V} L_{j,i}^* \big(\ind_{j=k}\,A_k+\sum_{\pi\in \cP(j,k)}L_\pi^* A_k L_\pi\big) L_{j,i}\Big)\otimes \ketbra jj.
    \end{align*} 
Since the set of paths obtained by concatenating one step from a given $i$ to a variable $j$, then some $\pi$ from $j$ to $k$, is exactly the set of paths from $i$ to $k$ of length $\geq 2$, and $(i,k)$ is the only path from $i$ to $k$ of length $1$, we obtain $\M^*(Z)=Z-A_k\otimes\ketbra kk$, so that $(\id-\M^*)(Z)=A$. If $Z'$ is another solution of \eqref{eq_pbDirichletV}, then $Y=Z'-Z$ satisfies $\M^*(Y)=Y$ and by the Perron-Frobenius Theorem of \cite{Gro} applied to the irreducible map~$\M^*$, we have $Y\in \cc\id_{\mathcal H}$. \qed

\paragraph{Proof of Proposition \ref{prop_sigmaDfinite}}
It is now a routine argument to construct $\mfP_{i,j}^D$ using the Banach-Steinhaus Theorem, as
\[\mfP_{j,i}^D(\rho)=\sum_{\pi\in \cP^{D}(i,j)} L_\pi \rho L_\pi^*.\]
One then has by definition $\pp_{i,\rho}(t_j\leq t_\bD<\infty)=\tr\, \mfP^D_{j,i}(\rho)$, and the second identity follows from $\pp_{i,\rho}(t_\bD<\infty)=\sum_{j\in \bD}\pp_{i,\rho}(t_j\leq t_\bD<\infty)$. Relations \eqref{eq_erhotj1} and \eqref{eq_erhotj2} are obtained as Equation \eqref{eq_esprhoti}.
\qed
\paragraph{Proof of Proposition \ref{prop_sortieps}}

We define
    \[p= \inf_{i\in D}\inf_{\rho \in\mathcal S(\mathfrak h_i)} \pp_{i,\rho}( t_{\bD} <  +\infty).\]
    We will show independently that $p>0$ and that $p\in\{0,1\}$, therefore proving Proposition \ref{prop_sortieps}.
    
    To prove that $p>0$, we use a simple adaptation of Lemma \ref{lemme_probaminoree}. Fix some~$\rho$ in $\mathcal S (\mathfrak h_i)$; there exist a unit vector $\varphi$ in~$\mathfrak h_i$ and $\lambda >0$ such that $\rho \geq \lambda \ketbra \varphi\varphi$. By irreducibility, for any~$j$ in $\bD$ there exists a path $\pi$ in $\cP(i,j)$ such that $L_\pi \,\varphi\neq 0$. There exists
$j'$ in $\bD$ (the first point of $\bD$ visited by the trajectory $\pi$) and a subpath $\pi'$ of $\pi$ belonging to $\cP^{D}(i,j')$, with necessarily $L_{\pi'} \,\varphi \neq 0$. We have shown $\tr\,\mfP^D_{i}(\rho)>0$ and, $\mfP^D_{i}$ being continuous, we have by a compactness argument that $\inf_{\rho\in \mathcal S(\h_i)}\tr\,\mfP^D_{i}(\rho)>0$, and therefore $p>0$ as $D$ is finite.
   
    We next prove that $p \in\{0,1\}$. By the strong Markov properrty, for any $n$ one has
    \begin{align*}
    1-p &=\sup_{i\in D}\sup_{\rho \in\mathcal S(\mathfrak h_i)} \pp_{i,\rho}(t_{\bD}=+\infty)\\
    &=  \sup_{i,\rho} \ee_{i,\rho} \big( \ind_{x_1,\ldots, x_n \in D} \, \pp_{x_n,\rho_n} (t_{\bD}=+\infty) \big)\\
    &\leq (1-p) \,  \pp_{i,\rho} ( x_1,\ldots, x_n \in D),
    \end{align*}
    and taking $n\to\infty$ leads to $(1-p)\leq (1-p)^2$, so that $p\in\{0,1\}$. This concludes our proof.
\qed

\paragraph{Proof of Lemma \ref{lemme_mfND}}

Let $j$ in $V$ with $\mathrm{dim}\,\h_j<\infty$. By irreducibility, there exists a path $\pi$ in $\mathcal P^D(j,k)$ for some $k\in \bD$ such that $\tr L_\pi \rho L_\pi^*\neq 0$. There exists $k'$ in $\bD$ and a subpath $\pi'$ of $\pi$ which belongs to $\mathcal P^{D\setminus\{j\}}(j,k')$ such that $\tr L_{\pi'} \rho L_{\pi'}^*\neq 0$, which implies that $\pp_{j,\rho}(t_j\leq t_\bD)~<~1$. In particular, $\tr\,\mfPD_{j,j}(\rho)<1$ for any $\rho$ in $\mathcal S(\h_j)$, so that $\|\mfPD_{j,j}\|<1$. The same discussion that allowed us to construct $\mfN_{j,i}$ shows that $\mfND_{j,i}$ is well-defined by $\mfND_{j,i}=(\id - \mfPD_{j,j})^{-1}\circ \mfPD_{j,i}$ and satisfies relations \eqref{eq_eeirho} and \eqref{eq_mfND}.
\qed

\paragraph{Proof of Proposition \ref{prop_pbDirichletD}}
By Lemma \ref{lemme_mfND}, all operators $\mfND_{j,i}$ and therefore the operator $Z$, are well-defined. Obviously $Z_j=B_j$ for $j\in \bD$; the proof that $(\id-\M^*)(Z)_i=A_i$ for $i\in D$ is similar to that for Proposition \ref{prop_probDirichletV}. Now consider two solutions $Z$ and $Z'$; then $Y=Z-Z'$ satisfies $Y_j=0$ for $j\in \bD$ and $(\id-\M^*)(Y)_i=0$ for $i\in D$. As in Lemma \ref{lemme_martingale} we can prove that, if $m_n=\big(\tr(\rho_n Y_{x_n})\big)_n$, then $m^D_n=m_{\inf(n,t_\bD)}$ is a $\pp_{i,\rho}$-martingale for any $i$ in $D$ and $\rho$ in $\mathcal S(\h_i)$. The optional sampling Theorem applied to the bounded martingale $\tr(\rho_n Y_{x_n})$ and the stopping time $t_\bD$ implies that 
 \[\tr\big(\rho \,Y_i\big)=\ee_{i,\rho}\big(\tr(\rho_{t_\bD} Y_{x_{t_\bD}}\big)=0.\]
Since this is true for any $\rho$ in $\mathcal S(\h_i)$, we obtain that $Y_i=0$, for any $i\in D$. \qed

\section{Proof for Section \ref{sec_varDirichlet}}

\paragraph{Proof of Lemma \ref{lemme_ptesQDB}}
Since $\|\M^*\|=1$, the quantum detailed balance condition implies that the spectrum of $\M^*$ is contained in $[-1,+1]$, so that $I-\M^*$ is a positive operator and $\Ecal(X)\geq0$ for all $X\in\Bcal(\Hcal)$. In addition, $\Ecal(X)=0$ if and only if $\M^*(X)=X$. If $\M$ is irreducible, which by the Perron-Frobenius Theorem for operators on a C*-algebra (see \cite{Gro}) applied to the irreducible map~$\M^*$, the identity  $\M^*(X)=X$ is equivalent with $X\in \cc\id_{\mathcal H}$. \qed

\paragraph{Proof of Theorem \ref{theo_VarDirichletpb}}
Let us write $Z=B+X+X'$ with $X\in \mathcal B(\mathcal H_D)$ and $X'\in \mathcal B(\mathcal H_{V\setminus(D\cup \bD)})$. By definition of $\bD$, one has $(\id-\M^*)(X')\in \mathcal B(\mathcal H_{V\setminus D})$. Denoting $C=(\id-\M^*)(B)$ we have that $Z$ is a solution of \eqref{eq_pbDirichletD} if and only if $(\id-\M^*)(X)_k=(A-C)_k$ for $k\in D$, or equivalently if 
\begin{equation} \label{eq_LaxMilgram}
\mathcal E(T,X)=\langle T, A-C\rangle_\diamond \quad \mbox{ for any } T\in\mathcal B(\mathcal H_D).
\end{equation}
By Lemma \ref{lemme_ptesQDB}, $\E(X,X)$ is non-negative and vanishes only if $X\in \cc \id_{\mathcal H}$. However, since $\bD\neq \emptyset$,  $\id_{\mathcal H}\not\in \mathcal B(\mathcal H_D)$ and one has $\mathcal E(X,X)>0$ for any $X\in \mathcal B(\mathcal H_D)$. Consequently, by a compactness argument, there exists $\lambda>0$ such that $\mathcal E(X,X)\geq \lambda\|X\|_\diamond^2$ for $X\in \mathcal B(\mathcal H_D)$.  One can then apply the Lax-Milgram Theorem (see \cite{Brezis}): there exists a unique $X_0$ satisfying \eqref{eq_LaxMilgram}, which in addition is the minimizer of 
\[\mathcal B(\mathcal H_D) \ni X\mapsto \frac12\,\mathcal E(X,X) - \langle X,A-C\rangle_\diamond = \frac12\,\mathcal E(X,X) + \mathcal E(X,B)- \langle X,A\rangle_\diamond.\]
The solutions of Equation \eqref{eq_pbDirichletD} are therefore the operators of the form  
\[Z=B+X_0+X'\]
for $X'\in\mathcal B(\mathcal H_{V\setminus(D\cup\bD)})$.
\qed 

\paragraph{Proof of Proposition \ref{prop_doublysto}}

The proof is simply a matter of computation.
For doubly stochastic OQW, $L_{ij}=L^*_{ji}$, the invariant state $\tau_\diamond$ is the identity and the Dirichlet form reads
\begin{equation*}
 \mathcal{E}(X)= \mathrm{Tr}\big(X^*(\mathrm{Id}-\mathfrak{M})X\big) =\sum_{i,j\in V}\mathrm{Tr}\big(X^*_i \delta_{ij} X_j - X_i^*L_{ij}X_jL_{ji}\big). 
\end{equation*}
On the other hand we have
\begin{eqnarray*}
\frac{1}{2} \|(\nabla X)\|_V^2 &=& \frac{1}{2} \sum_{i,j\in V} \mathrm{Tr}\big((X_iL_{ij}-L_{ij}X_j)(L_{ji}X_i^*-X_j^*L_{ji})\big)\\
&=& \frac{1}{2}\sum_{i,j\in V} \mathrm{Tr}\big( X_iL_{ij}L_{ji}X_i^* +X^*_iL_{ij}L_{ji}X_i - 2 L_{ij}X_jL_{ji}X_i^*\big).
\end{eqnarray*}
The two formulas coincide since $\sum_{j\in V} L_{ij}L_{ji}=\mathrm{Id}$ for doubly stochastic OQW.
\qed
\end{document}